\documentclass[a4paper,12pt]{article}
\usepackage{amsmath}
\usepackage{amsthm}
\usepackage{amsfonts}
\usepackage{amssymb}
\usepackage{mathrsfs}
\usepackage{graphicx}
\usepackage{color}
\usepackage[colorlinks=true, allcolors=blue]{hyperref}

\newtheorem{theorem}{Theorem}
\newtheorem{lemma}{Lemma}
\newtheorem{proposition}{Proposition}

\usepackage[english]{babel}

\numberwithin{equation}{section}

\hypersetup{ pdftitle={Draft},
	pdfauthor={Mirda Prisma Wijayanto, Fiki Taufik Akbar, Bobby Eka Gunara}, 
	pdfkeywords={Global Esistence, Einstein Klein Gordon, Higher dimensions}, 
	bookmarksnumbered, pdfstartview={FitH}, urlcolor=blue, }

\begin{document}
	\pagestyle{plain}
	
	
	
	
	\title{\LARGE\textbf{Global Existence and Completeness of Classical Solutions in Higher Dimensional Einstein-Klein-Gordon System}}

	\author{Mirda Prisma Wijayanto, Fiki Taufik Akbar, Bobby Eka Gunara\footnote{Corresponding author}\\ 
		\\
		\textit{\small Theoretical High Energy Physics Research Division, }\\
		\textit{\small Faculty of Mathematics and Natural Sciences,}\\
		\textit{\small Institut Teknologi Bandung}\\
		\textit{\small Jl. Ganesha no. 10 Bandung, Indonesia, 40132}
		\\  \\
		\\
		\small email: mirda.prisma.wijayanto@students.itb.ac.id, ftakbar@itb.ac.id, bobby@itb.ac.id}
	
	\date{\today}
	
	\maketitle
	
	
	

	\begin{abstract}
		In this paper we study the global existence and completeness of classical solutions of gravity coupled a scalar field system called Einstein-Klein-Gordon system in higher dimensions. We introduce a new ansatz function to reduce the problem into a single first-order integro-differential equation. Then, we employ the contraction mapping in the appropriate Banach space. Using Banach fixed theorem, we show that there exists a unique fixed point, which is the solution of the theory. For a given initial data, we prove the existence of both local and global classical solutions.  We also study the completeness properties of the spacetime. Here, we introduce a mass-like function for $D\geq 4$ in Bondi coordinates. The completeness of spacetime along the future directed timelike lines outward to a region which resembles the event horizon of the black hole.
	\end{abstract}
	
	%
	%
	%
	%
	%

	\section{Introduction}	
	Einstein's equation is one of the fundamental equations in physics. In General Relativity, this equation explains the dynamics of the spacetime geometry and the matter distribution within. Moreover, this equation is nothing but belongs to a class of nonlinear partial differential equation. Solving analytically  this equation remains  open  for the past few decades.
	
	Apart from the above problem, another attempt is to study the Cauchy problem for the Einstein equations. This problem was firstly started by Choquet-Bruhat \cite{Choquet}. The existence of local solutions to the Einstein equations has been proved \cite{Chrusciel,Friedrich1} using the harmonic coordinates. One of the most remarkable result is the global existence solutions to the vacuum Einstein equations for initial data close to the trivial one due to Christodoulou and Kleinerman \cite{ChrisKlein}. An extension of the stability theorem of the Minkowski space has been discovered by Bieri \cite{Bieri}.
	
	It is of interest to extend the above study to the case of gravity-scalar system because it describes an important toy model of gravitational collapse in general relativity. The global initial value problem for Einstein equations in the spherically symmetric case with a massless scalar field as the material model for the sufficiently small initial data was initiated by Christodoulou \cite{Chris1} (see \cite{Chris2} for the large initial data). Then, Christodoulou \cite{Chris3} prove that the self-similar solution of this system generate singularities in finite time for certain initial conditions. It was later found in \cite{Chris4} that there are instabilities in the naked singularity formation, and it is consistent with the cosmic censorship conjecture problem of this model. Malec \cite{Malec} studied the coupled Einstein and scalar system where Cauchy data is given on spacelike hypersurface, as well as by examining the local existence of the Cauchy problem outside outgoing null hypersurfaces. Furthermore, the global initial value problem of spherically symmetric solutions for Einstein-Scalar system or also known as Einstein-Klein-Gordon system has been studied by Chae \cite{Chae} and Wijayanto et al. \cite{Prisma} in four dimensions. There are also some related works \cite{Dafermos}-\cite{Luk} and references therein.
	
	\subsection{Motivation and Challenges}
	There is increasing interest in the generalizations of Einstein's general theory of relativity in higher dimensions \cite{Clifton}-\cite{Coley}. In fact, theories that attempt to unify forces, e.g. Kaluza-Klein theory, string theory, and superstring theory, often require higher spacetime dimensions to explain certain phenomena in physics, and often have different properties from four-dimensional theories. One of the differences can be seen when we study the asymptotic behavior of solutions to Einstein equations. The decay properties of the nontrivial vacuum solution in Minkowski spacetime near spatial infinity is of order $1/r^{D-3}$ as $r\rightarrow \infty$ for the fixed $t$. This behavior can be observed explicitly in the Schwarzschild metric, and corresponds to the decay of the point mass potential in Newtonian gravity. On the other hand, the decay properties from Minkowski spacetime near null infinity is of order $1/r^{(D-2)/2}$ as $r\rightarrow \infty$ for the fixed $u$. This behavior can be observed from the phenomenon of graviton propagation in Minkowski space. We can verify that for $D=4$ both decays have the same order. However, for $D>4$, the decay at null infinity is slower than the decay at the spatial infinity \cite{Hollands}. This result enables the development of new concepts and possibilities when we work in higher dimensions. Motivated by these developments, we are interested in extending the study of the global initial value problem of Einstein-Klein-Gordon system in higher dimensions ($D\geq 4$).
	
	In the present work, we also study the completeness properties of the spacetime. The nature of spacetime completeness is one of the open problems in General Relativity formulated by the cosmic censorship conjecture: there exists a solution of the Einstein equations for a given arbitrary asymptotically flat initial data which is a globally hyperbolic spacetime possessing a complete future null infinity. In the four dimensional case, Christodoulou \cite{Chris1} shows that the Einstein-Klein-Gordon system with $V=0$ for small initial data has a global solution in the future complete spacetime for $r_0>2\tilde{M}$, where $r_0$ represents the timelike line and $\tilde{M}$ is the Bondi mass in the four dimensions. As a novelty, we introduce a Bondi mass-like function $M$ for $D\geq 4$. 
	In order for the function $M$ to be defined as mass physically, it must satisfy the property of mass in General Relativity, namely the Positive Mass Theorem which states that for a nontrivial isolated physical system, the total energy (mass), which includes contributions from matter and gravity, is positive. This theorem has been proven by \cite{Schoen} for ADM mass defined at spatial infinity and by \cite{Ludvig} for Bondi mass defined at null infinity. According to the Positive Mass Theorem, we prove that $M$ is a monotonically nondecreasing function of $r$ at each $u$, and also a monotonically nonincreasing function of $u$. Finally, we show that the completeness of spacetime along the future directed timelike lines outward to a region which resembles the event horizon of the black hole.
	
	Let us introduce the radial coordinate $r$. We set that the central world line is in $r=0$, and the world lines $r=r_0$ in each half-plane are all timelike. We also introduce a timelike coordinate $u$ with the properties that $u$ is constant on the future light cone of each point on the central world line. On a world line $r=r_o$, $u$ tends to the proper time as $r_0\rightarrow \infty$. Hence, the metric of our spacetime can be written in the form
	\begin{eqnarray}\label{metric}
	\mathrm{d}s^2 = - e^{2F(u,r)}\mathrm{d}u^2 - 2e^{F(u,r)+G(u,r)}\mathrm{d}u\mathrm{d}r + r^2 \mathrm{d}\Omega^{D-2}.
	\end{eqnarray}
	The spatial metric $\mathrm{d}\Omega^{D-2}$ describes ($D-2$)-spatial compact manifold. In particular, if the $(D-2)$-spatial geometry is sphere, then the metric (\ref{metric}) becomes flat as $r$ goes to infinity. We define the functions $F$ and $G$ which tend to 0 as $r$ goes to infinity at each $u$  according to asymptotic Ricci flatness. Such coordinate system is called the Bondi coordinate system. Since we only consider the future of some initial future light cone with vertex at the center, the range of the coordinates $u$ and $r$ are $0\leq u<\infty$ and $0\leq r<\infty$ respectively. Here, we define $u=0$ as the initial future light cone. 
	
	The action of our system has the form
	\begin{eqnarray}\label{action}
	S=\int d^Dx \sqrt{-\det({g_{\mu\nu}})}\left\{\frac{1}{2}\partial_\mu\phi\partial^\mu\phi - V(\phi)\right\},
	\end{eqnarray}
	where $\det({g_{\mu\nu}})$ is the determinant of the metric (\ref{metric}), and $\phi$ is the real scalar field.  Let us denote the scalar potential $V(\phi)$ fulfils the estimate
	\begin{equation}\label{estimasi Potensial}
	|V(\phi)|+\left|\frac{\partial V(\phi)}{\partial \phi}\right||\phi|+\left|\frac{\partial^2V(\phi)}{\partial\phi^2}\right||\phi|^2\leq K_0 |\phi|^{p+1},
	\end{equation}
	where $K_0\geq 0$, and $p\in \mathbb{R}^+$. The values of $p$ are related to the decay of the order $k$ of the solutions. In the previous work \cite{Prisma}, we studied the decay properties of the four-dimensional Einstein-Klein-Gordon system for $k\in[3,\infty)$. In this paper, we extend this problem to higher dimensions, where the decay order depends on the spacetime dimension $D$ such that $k>\frac{D}{2}$, where $k\in\mathbb{R}^+$. Here, we also choose that $V(\phi)$ is a negative function.
	
	In the case of higher dimensions, we obtain the complete equations of motion of the Einstein-Klein-Gordon system in the following form
	\begin{align}
	\frac{\partial}{\partial u}\left[r\frac{\partial\phi}{\partial r}+\frac{(D-2)}{2}\phi\right]=&\frac{\tilde{g}}{2}\frac{\partial}{\partial r}\left[r\frac{\partial\phi}{\partial r}+\frac{(D-2)}{2}\phi\right]+\frac{1}{2}\left[\frac{\hat{R}(\hat{\sigma})}{(D-2)}g-\frac{(D-2)}{2}\tilde{g}\right]\frac{\partial\phi}{\partial r}\nonumber\\
	&+\frac{8\pi gr^2}{(D-2)}\frac{\partial\phi}{\partial r}V(\phi) + \frac{gr}{2}\frac{\partial V(\phi)}{\partial \phi},\nonumber
	\end{align}
	where $\tilde{g}=e^{F-G}$ and $g=e^{F+G}$. In the starting point, we simplify the above equation by defining a new ansatz function that is more general than previous works by \cite{Chris1,Chae,Prisma} as follows
	\begin{eqnarray}\label{h}
	h=r\frac{\partial\phi}{\partial r}+\frac{(D-2)}{2}\phi,
	\end{eqnarray}
	and
	\begin{eqnarray}\label{htilde}
	\phi :=\tilde{h}= r^{-\frac{(D-2)}{2}}\int_{0}^{r}hs^{\frac{D-4}{2}}\mathrm{d}s.
	\end{eqnarray}
	To proceed, we also define
	\begin{eqnarray}\label{g}
	g:=e^{F+G}=\exp \left[-\int_{r}^{\infty}\frac{8\pi}{(D-2)s}\left(h-\frac{(D-2)}{2}\tilde{h}\right)^2\mathrm{d}s\right],
	\end{eqnarray}
	and
	\begin{align}\label{gtilde}
	\tilde{g}:= e^{F-G}=\frac{\hat{R}(\hat{\sigma})}{(D-2)}\bar{g}-\frac{(D-4)}{(D-2)}\frac{\hat{R}(\hat{\sigma})}{r^{D-3}}\int_{0}^{r}\bar{g}s^{D-4}\mathrm{d}s+\frac{16\pi  }{(D-2)}\frac{1}{r^{D-3}}\int_{0}^{r} gs^{D-2}V(\tilde{h})\mathrm{d}s,
	\end{align}
	where $\bar{g}:=\bar{g}(u,r)=\frac{1}{r}\int_{0}^{r}g(u,s)\mathrm{d}s$. We denote $\hat{R}(\hat{\sigma})=\hat{\sigma}^{ij}\hat{R}_{ij}(\hat{\sigma})$ as the $(D-2)$-spatial Ricci scalar of compact manifold. In order to have a consistent picture, we assume that the $(D-2)$-spatial geometry has to be Einstein implying that the Ricci scalar $\hat{R}(\hat{\sigma})$ is constant. Furthermore, we require that $\hat{R}(\hat{\sigma})$ must also be positive according to the completeness properties of the spacetime. Another challenge we found was that the second term of the right-hand side of the equation (\ref{gtilde}) only appears in higher dimensions. To overcome this problem, we have to calculate additional estimates for all functions containing $\tilde{g}$. This is of course requires some new mathematical ideas especially when we study the properties of the mass-like functions, which are then used to analyze the completeness of spacetime in higher dimensions ($D\geq 4$).
	
	\subsection{Structure of the Paper}
	The structure of this paper is the following. In section \ref{sec2}, we construct the Einstein-scalar system in higher dimensions ($D\geq 4$). We use the coordinate system according to metric (\ref{metric}). Then, we construct the Einstein equation and the equation of motion respectively. In the last part, we use the ansatz function (\ref{h}) and (\ref{htilde}) to reduce the equation of motion into a single first-order evolution equation as the main equation of this problem. In section \ref{sec3}, we give the setup to prove the existence of classical solution consist of local and global existence. We employ the contraction mapping principle in the appropriate Banach space. This gives some arguments to prove the existence of the global solution in view of Banach fixed theorem. In the end of this section, we state the main theorem of global existence. In section \ref{sec4}, we study the completeness properties of the spacetime along the timelike lines. Finally, we write down the detailed calculation in the Appendix.
	
	\section{Higher Dimensional Einstein-Scalar System}
	\label{sec2}
	\subsection{Einstein Equation}
	It is of interest to write down the Ricci tensor related to metric (\ref{metric})
	\begin{eqnarray}
	R_{00}&=&-e^{F-G}\left[\partial_1\partial_0 (F+G) +\frac{(D-2)}{r}\partial_0 G\right]\nonumber\\
	&&+ e^{2(F-G)}\left[(\partial_1F)^2+\partial_1(\partial_1 F)-\partial_1F\partial_1 G +\frac{(D-2)}{r}\partial_1 F \right]\\
	R_{01}&=&R_{10}=-\partial_1\partial_0(F+G)\nonumber\\
	&&+e^{F-G}\partial_1(F-G)\partial_1 F + e^{F-G}\partial_1(\partial_1 F)+\frac{(D-2)}{r}e^{F-G}\partial_1 F\\
	R_{11}&=&\frac{(D-2)}{r}\partial_1(F+G)\\
	R_{ij}&=& \hat{R}_{ij}(\hat{\sigma}) - e^{-2G}\left[r\partial_1(F-G)+D-3\right]\hat{\sigma}_{ij}
	\end{eqnarray}
	where we have denoted $\partial_0=\frac{\partial}{\partial u}$, $\partial_1=\frac{\partial}{\partial r}$, $\hat{R}_{ij}(\hat{\sigma})$ is the $(D-2)$-spatial Ricci tensor of compact manifold, and $i,j=1,2,...,D-1$. Furthermore, we also have the scalar curvature
	\begin{eqnarray}
	R&=&\frac{1}{r^2}\hat{R}(\hat{\sigma})+2e^{-(F+G)}\partial_1\partial_0 (F+G)+e^{-2G}\left[-2(\partial_1 F)^2-2\partial_1(\partial_1 F)\right.\nonumber\\
	&&\left.+2(\partial_1G)(\partial_1 F)-2\frac{(D-2)}{r}\partial_1(F-G)-\frac{(D-2)(D-3)}{r^2}\right]
	\end{eqnarray}
		
	Now, we write the Einstein equation in higher dimensions
	\begin{eqnarray}\label{Einstein}
	R_{\mu\nu}=8\pi \left(T_{\mu\nu}-\frac{1}{(D-2)}g_{\mu\nu} T\right),
	\end{eqnarray}
	with the energy-momentum tensor
	\begin{eqnarray}\label{Energy-momentum tensor}
	T_{\mu\nu}=\partial_\mu\phi\partial_\nu \phi-\frac{1}{2}g_{\mu\nu}g^{\alpha\beta}\partial_\alpha\phi\partial_\beta\phi + g_{\mu\nu}V(\phi),
	\end{eqnarray}
	where we have denoted $T$ as the trace of $T_{\mu\nu}$.
	
	The $\{rr\}$ and $\{ij\}$ components of (\ref{Einstein}) satisfy
	\begin{eqnarray} 
	\partial_1(F+G)&=&\frac{8\pi r}{D-2}\left(\frac{\partial \phi}{\partial r}\right)^2,\label{rr}\\
	\partial_1(F-G)\hat{\sigma}_{ij}&=&-\frac{1}{r}\left[-e^{2G}\hat{R}_{ij}(\hat{\sigma})+(D-3)\hat{\sigma}_{ij}\right]+\frac{16\pi r e^{2G}}{(D-2)}V(\phi)\hat{\sigma}_{ij}\label{ij}
	\end{eqnarray}
	respectively. Additionally, the combination of the $\{00\}$ and $\{01\}$ components of (\ref{Einstein}) yields
	\begin{align} \label{00-01}
	\frac{(D-2)}{r}\frac{\partial G}{\partial u}=8\pi\left[\frac{\partial\phi}{\partial u}\frac{\partial\phi}{\partial r}-e^{-(F-G)}\left(\frac{\partial\phi}{\partial u}\right)^2\right].
	\end{align}
	As we can see from (\ref{rr}), (\ref{ij}), and (\ref{00-01}), there are no terms that contain the second derivative of $u$ and $r$ in $F $ and $G$. Thus, the metric functions $F$ and $G$ are no longer dynamical fields, but they simply become constraints. In the rest of the paper we will use (\ref{rr}) and (\ref{ij}) for our analysis rather than other components since it is easier and already common in the standard literature \cite{Chris1, Chae, Prisma}.
	
	\subsection{Equation of Motion}
	From the action (\ref{action}), we obtain the equation of motion
	\begin{eqnarray}\label{scalar}
	\Box \phi = -\frac{\partial V(\phi)}{\partial \phi},
	\end{eqnarray}
	where $\Box$ is defined as the d'Alembert operator with respect to the metric (\ref{metric}) as follows 
	\begin{eqnarray}\label{EoM}
	\Box\phi &=& -2 e^{-(F+G)}\left[\frac{\partial^2\phi}{\partial u \partial r}+\frac{(D-2)}{2r}\frac{\partial\phi}{\partial u}\right]\nonumber\\ &&+e^{-2G}\left[\frac{\partial^2\phi}{\partial r^2}+\frac{\partial\phi}{\partial r}\left(\frac{(D-2)}{r}+\frac{\partial F}{\partial r}-\frac{\partial G}{\partial r}\right)\right].
	\end{eqnarray}
		
	Now, we introduce a new operator
	\begin{eqnarray}
	\mathcal{D}=\frac{\partial}{\partial u}-\frac{\tilde{g}}{2}\frac{\partial}{\partial r},
	\end{eqnarray}
	describes the derivative along the incoming light rays parametrized by $u$. 
	Thus, we write (\ref{scalar}) as the single first-order integro-differential equation as follows
	\begin{eqnarray}\label{Dh higher}
	\mathcal{D}h&=&\frac{1}{2r}\left[\frac{\hat{R}(\hat{\sigma})}{(D-2)}g-\frac{(D-2)}{2}\tilde{g}\right]\left[h-\frac{(D-2)}{2}\tilde{h}\right]\nonumber\\
	&&+\frac{8\pi gr}{(D-2)}\left[h-\frac{(D-2)}{2}\tilde{h}\right]V(\tilde{h})+ \frac{gr}{2}\frac{\partial V(\tilde{h})}{\partial \tilde{h}},
	\end{eqnarray}
	where $g$ and $\tilde{g}$ are define in (\ref{g}) and (\ref{gtilde}) respectively. Equation (\ref{Dh higher}) provides the nonlinear evolution equation with respect to $h$ in higher dimensions.
	
	\section{Existence of Classical Solution}
	\label{sec3}
	Let us define a map $h\mapsto \mathcal{F}(h)$ as the solution of equation (\ref{Dh higher}), with the initial condition $\mathcal{F}(h)(0,r)=h(0,r)$.	Let $r(u)=\chi(u;r_0)$ be the solution of the ordinary differential equation
	\begin{equation} \label{PDB}
	\frac{\mathrm{d}r}{\mathrm{d}u}=-\frac{1}{2}\tilde{g}(u,r),~~~r(0)=r_0.
	\end{equation}
	We introduce the characteristic function $r_1=\chi(u_1;r_0)$. Then, from (\ref{PDB}) we obtain
	\begin{equation}\label{kondisi awal r}
	r_1 = r_0 - \frac{1}{2}\int_{0}^{u_1} \tilde{g}(u,\chi(u;r_0))~\mathrm{d}u.
	\end{equation}
	Therefore, we can write the \textcolor{blue}{differential} equation (\ref{Dh higher}) as the integral equation
	\begin{eqnarray}\label{Fcurl}
	&&\mathcal{F}(u_1,r_1)=h(0,r_0)\exp\left\{\int_{0}^{u_1} \left(\frac{1}{2r}\left[\frac{\hat{R}(\hat{\sigma})}{(D-2)}g-\frac{(D-2)}{2}\tilde{g}\right]+\frac{8\pi g r}{(D-2)}V(\tilde{h})\right)_\chi\mathrm{d}u'\right\}\nonumber\\
	&&+\int_{0}^{u_1} \exp\left\{\int_{u}^{u_1} \left(\frac{1}{2r}\left[\frac{\hat{R}(\hat{\sigma})}{(D-2)}g-\frac{(D-2)}{2}\tilde{g}\right]+\frac{8\pi g r}{(D-2)}V(\tilde{h})\right)_\chi\mathrm{d}u'\right\}[f]_\chi\mathrm{d}u,
	\end{eqnarray}
	where
	\begin{eqnarray}\label{f prisma}
	f=-\left(\frac{1}{2r}\left[\frac{\hat{R}(\hat{\sigma})}{(D-2)}g-\frac{(D-2)}{2}\tilde{g}\right]+\frac{8\pi g r}{(D-2)}V(\tilde{h})\right)\frac{(D-2)}{2}\tilde{h}+\frac{gr}{2}\frac{\partial V(\tilde{h})}{\partial \tilde{h}}.
	\end{eqnarray}
	
	Now, let us define
	\begin{equation}
	\mathcal{G}(u,r)=\frac{\partial \mathcal{F}}{\partial r}(u,r),
	\end{equation}
	satisfies
	\begin{align}\label{eq.G}
	\mathcal{D}\mathcal{G}=&\left[\frac{1}{2}\frac{\partial \tilde{g}}{\partial r}+\frac{1}{2r} \left(\frac{\hat{R}(\hat{\sigma})}{(D-2)}g-\frac{(D-2)}{2}\tilde{g}\right)+\frac{8\pi g r}{(D-2)}V(\tilde{h})\right]\mathcal{G}\nonumber\\
	& + \left[\frac{1}{2r}\frac{\partial}{\partial r}\left(\frac{\hat{R}(\hat{\sigma})}{(D-2)}g-\frac{(D-2)}{2}\tilde{g}\right)-\frac{1}{2r^2}\left(\frac{\hat{R}(\hat{\sigma})}{(D-2)}g-\frac{(D-2)}{2}\tilde{g}\right)\right.\nonumber\\
	&\left.+\frac{8\pi g r }{(D-2)} \frac{\partial V(\tilde{h})}{\partial \tilde{h}} \frac{\partial \tilde{h}}{\partial r}+\frac{8\pi g V(\tilde{h})}{(D-2)}+\frac{8\pi r V(\tilde{h})}{(D-2)}\frac{\partial g}{\partial r}\right]\left(\mathcal{F}-\frac{(D-2)}{2}\tilde{h}\right)\nonumber\\
	& -\left[\frac{(D-2)}{2}\left(\frac{1}{2r} \left(\frac{\hat{R}(\hat{\sigma})}{(D-2)}g-\frac{(D-2)}{2}\tilde{g}\right)+\frac{8\pi g r V(\tilde{h})}{(D-2)}\right)-\frac{g r }{2}\frac{\partial^2V(\tilde{h})}{\partial\tilde{h}^2}\right]\frac{\partial \tilde{h}}{\partial r}\nonumber\\
	&+\left[\frac{r}{2}\frac{\partial g}{\partial r}+\frac{g}{2}\right]\frac{\partial V(\tilde{h})}{\partial\tilde{h}},
	\end{align}
	with the initial condition $\mathcal{G}(0,r_0)=\frac{\partial h}{\partial r}(0,r_0)$. Using the characteristics (\ref{kondisi awal r}), we write the equation (\ref{eq.G}) as the integral equation
	\begin{align}\label{Gcurl}
	&\mathcal{G}(u_1,r_1)=\frac{\partial h}{\partial r}(0,r_0)\exp\left\{\int_{0}^{u_1}\left[\frac{1}{2}\frac{\partial \tilde{g}}{\partial r}+\frac{1}{2r} \left(\frac{\hat{R}(\hat{\sigma})}{(D-2)}g-\frac{(D-2)}{2}\tilde{g}\right)+\frac{8\pi g rV(\tilde{h})}{(D-2)}\right]_\chi\mathrm{d}u \right\}\nonumber\\
	&+\int_{0}^{u_1}\exp\left\{\int_{u}^{u_1}\left[\frac{1}{2}\frac{\partial \tilde{g}}{\partial r}+\frac{1}{2r} \left(\frac{\hat{R}(\hat{\sigma})}{(D-2)}g-\frac{(D-2)}{2}\tilde{g}\right)+\frac{8\pi g rV(\tilde{h})}{(D-2)}\right]_\chi\mathrm{d}u' \right\}[f_1]_\chi~\mathrm{d}u,
	\end{align}
	where
	\begin{align} \label{f1}
	f_1 =& \left[\frac{1}{2r}\frac{\partial}{\partial r}\left(\frac{\hat{R}(\hat{\sigma})}{(D-2)}g-\frac{(D-2)}{2}\tilde{g}\right)-\frac{1}{2r^2}\left(\frac{\hat{R}(\hat{\sigma})}{(D-2)}g-\frac{(D-2)}{2}\tilde{g}\right)+\frac{8\pi g r }{(D-2)} \frac{\partial V(\tilde{h})}{\partial \tilde{h}} \frac{\partial \tilde{h}}{\partial r}\right.\nonumber\\
	&\left.+\frac{8\pi g V(\tilde{h})}{(D-2)}+\frac{8\pi r V(\tilde{h})}{(D-2)}\frac{\partial g}{\partial r}\right]\left(\mathcal{F}-\frac{(D-2)}{2}\tilde{h}\right)+\left[\frac{r}{2}\frac{\partial g}{\partial r}+\frac{g}{2}\right]\frac{\partial V(\tilde{h})}{\partial\tilde{h}}\nonumber\\
	& -\left[\frac{(D-2)}{2}\left(\frac{1}{2r} \left(\frac{\hat{R}(\hat{\sigma})}{(D-2)}g-\frac{(D-2)}{2}\tilde{g}\right)+\frac{8\pi g r V(\tilde{h})}{(D-2)}\right)-\frac{g r }{2}\frac{\partial^2V(\tilde{h})}{\partial\tilde{h}^2}\right]\frac{\partial \tilde{h}}{\partial r}.
	\end{align}
	
	Suppose that equation (\ref{Dh higher}) has two different solutions namely $h_1$ and $h_2$, with  $h_1(0,r_1)=h_2(0,r_2)$. Let us define  $\Theta=\mathcal{F}(h_1)-\mathcal{F}(h_2)$, and $g_l=g(h_l)$, $\mathcal{F}_l=\mathcal{F}(h_l)$, $\mathcal{G}_l=\mathcal{G}(h_l)$ for $l=1,2$. 	As previously, we define the characteristic $\chi_l=\chi_l(u,r)$ as follows
	\begin{align}
	\frac{\mathrm{d}r}{\mathrm{d}u}=-\frac{\tilde{g}_1}{2}(\chi_1(u,r),u);~~~r(0)=r_0.
	\end{align}
	Then, we write $\Theta$ as the integral equation
	\begin{align}\label{Theta}
	\Theta(u_1,r_1)=\int_{0}^{u_1} \exp\left\{\int_{u}^{u_1} \left[\frac{1}{2r}\left(\frac{\hat{R}(\hat{\sigma})}{(D-2)}g_1-\frac{(D-2)}{2}\tilde{g}_1\right)+\frac{8\pi r g_2V(\tilde{h}_1)}{(D-2)}\right]_{\chi_1}~\mathrm{d}u'\right\}[\tilde{\psi}]_{\chi_1}\mathrm{d}u\;,
	\end{align}
	where 
	\begin{eqnarray} \label{psi}
	\tilde{\psi}&=&\frac{1}{2}\left(\tilde{g}_1-\tilde{g}_2\right)\mathcal{G}_2-\frac{1}{2r}\left[\frac{\hat{R}(\hat{\sigma})}{(D-2)}g_1-\frac{(D-2)}{2}\tilde{g}_1\right]\frac{(D-2)}{2}\left(\tilde{h}_1 -\tilde{h}_2\right)\nonumber\\
	&&+\frac{1}{2r}\left(\frac{\hat{R}(\hat{\sigma})}{(D-2)}(g_1-g_2)-\frac{(D-2)}{2}(\tilde{g}_1-\tilde{g}_2)\right)\mathcal{F}_2+\frac{8\pi r g_2}{(D-2)}\left(V(\tilde{h}_1)-V(\tilde{h}_2)\right)\mathcal{F}_2\nonumber\\
	&&-\frac{1}{2r}\left(\frac{\hat{R}(\hat{\sigma})}{(D-2)}(g_1-g_2)-\frac{(D-2)}{2}(\tilde{g}_1-\tilde{g}_2)\right)\frac{(D-2)}{2}\tilde{h}_2\nonumber\\
	&&+\frac{8\pi r}{(D-2)}(g_1-g_2)\mathcal{F}_1V(\tilde{h}_1)-4\pi r(g_1-g_2)\tilde{h}_1 V(\tilde{h}_1)\nonumber\\
	&&-4\pi r g_2(\tilde{h}_1 - \tilde{h}_2)V(\tilde{h}_2)-4\pi r g_2 \tilde{h}_1\left(V(\tilde{h}_1)-V(\tilde{h}_2)\right)+\frac{r}{2}(g_1-g_2)\tilde{h}_1\frac{\partial^2 V(\tilde{h}_1)}{\partial\tilde{h}_1^2}\nonumber\\
	&&+\frac{r}{2}g_2(\tilde{h}_1-\tilde{h}_2)\frac{\partial^2 V(\tilde{h}_1)}{\partial\tilde{h}_1^2} +\frac{r}{2}g_2\tilde{h}_1\left(\frac{\partial^2 V(\tilde{h}_1)}{\partial\tilde{h}_1^2}-\frac{\partial^2 V(\tilde{h}_2)}{\partial\tilde{h}_2^2}\right).
	\end{eqnarray}
	
	\subsection{Local Existence}
	This section is devoted to prove the existence of local classical solution to (\ref{Dh higher}). For the starting point, let us define a function space
	\begin{eqnarray}\label{space Xhat}
	\hat{X} =\{h\in C^1([0,u_0] \times [0,\infty) )\; | \;\|h\|_{\hat{X}} < \infty\}\:, 
	\end{eqnarray}
	equipped with the norm
	\begin{eqnarray} \label{hx hat}
	\|h\|_{\hat{X}} =\sup_{u\in[0,u_0]} \sup_{r\geq 0}\left\{ (1+r+u)^{k-1}|h(u,r)| + (1+r+u)^{k}\left|\frac{\partial h}{\partial r}(u,r)\right|  \right\}\:, 
	\end{eqnarray}
	with $k>\frac{D}{2}$, where $k\in \mathbb{R}^+$, and $D\geq 4$. Then, we introduce
	\begin{eqnarray}
	\hat{X}_0=\{h\in C^1([0,u_0])\; | \;\|h\|_{\hat{X}_0} < \infty\}\:,
	\end{eqnarray}
	with its norm
	\begin{eqnarray}
	\|h\|_{\hat{X}_0} =\sup_{r\geq 0}\left\{ (1+r)^{k-1}|h(r)| + (1+r)^{k}\left|\frac{\partial h}{\partial r}(r)\right|  \right\}\:.
	\end{eqnarray}
	We also introduce the space $\hat{Y}$ containing $\hat{X}$ as follows
	\begin{eqnarray}\label{space Yhat}
	\hat{Y}=\{h\in C^1([0,u_0] \times [0,\infty) )\; | \; h(0,r) = h_0(r), \|h\|_{\hat{Y}} < \infty\}\:,
	\end{eqnarray}
	with
	\begin{eqnarray}\label{norm Yhat}
	\|h\|_{\hat{Y}} =\sup_{u\in[0,u_0]} \sup_{r\geq 0}\left\{ (1+r+u)^{k-1}|h(u,r)|\right\}\:,
	\end{eqnarray}
	The space function $\hat{X}, \hat{X}_0,$ and $\hat{Y}$ are Banach spaces. Let us denote $\|h\|_{\hat{X}}=\hat{x}$, $\|h(0,.)\|_{\hat{X}_0}=\hat{d}$, and $\|h_1-h_2\|_{\hat{Y}}=\hat{y}$.
	
	\begin{lemma}
	Let $D\geq 4$. For any $\mathcal{L}_1(\hat{x})>\hat{d}$, there exists $\delta(\hat{x},\hat{d})>0$ such that if $u_0<\delta$, the map $\mathcal{F}$ defined in (\ref{Fcurl}) is contained in the closed ball of radius $\hat{x}$ in the space $\hat{X}$. 
	\end{lemma}
	\begin{proof}
		Using the definition $\|h\|_{\hat{X}}=\hat{x}$, we obtain the estimate for (\ref{htilde}) as follows
		\begin{eqnarray}\label{htilde local}
		|\tilde{h}|&\leq& r^{-\frac{(D-2)}{2}}\int_{0}^{r}\frac{\|h\|_{\hat{X}}}{(1+s+u)^{k-1}}s^{\frac{D-4}{2}}\mathrm{d}s\nonumber\\
		&\leq&\frac{2\hat{x}}{(2k-D)}\frac{1}{(1+u)^{\frac{(2k-D)}{2}}(1+r+u)^{\frac{(D-2)}{2}}}.
		\end{eqnarray}
		Thus,
		\begin{eqnarray}\label{h-htilde}
		\left|h-\frac{(D-2)}{2}\tilde{h}\right|\leq |h|+\frac{(D-2)}{2}|\tilde{h}|\leq \frac{C\hat{x}}{(1+u)^{\frac{(2k-D)}{2}}(1+r+u)^{\frac{(D-2)}{2}}},
		\end{eqnarray}
		for any $C:=C(k,D)>0$ depends on $k$ and dimensions $D$. 
		
		We calculate
		\begin{eqnarray}
		\int_{r}^{\infty}\frac{1}{s}\left|h-\frac{(D-2)}{2}\tilde{h}\right|^2\mathrm{d}s\leq \frac{\hat{x}^2}{(D-2)(1+u)^{2k-D}(1+r+u)^{D-2}},
		\end{eqnarray}
		such that from (\ref{g}) we have
		\begin{eqnarray}\label{estimate g}
		|g(u,r)| \geq \exp \left[-\frac{8\pi \hat{x}^2}{(D-2)^2(1+u)^{2k-D}(1+r+u)^{D-2}}\right].
		\end{eqnarray}
		Then, from (\ref{h-htilde}) we obtain
		\begin{align}
		|g(u,r)-g(u,r')|\leq&\int_{r'}^{r}\left|\frac{\partial g}{\partial s}(u,s) \right|\mathrm{d}s\leq \frac{8\pi}{(D-2)} \int_{r'}^{r}\frac{1}{s}\left|h-\frac{(D-2)}{2}\tilde{h}\right|^2\mathrm{d}s\nonumber\\
		=& \frac{8\pi}{(D-2)^2}\frac{\hat{x}^2}{(1+u)^{2k-D}}\left[\frac{1}{(1+r'+u)^{D-2}}-\frac{1}{(1+r+u)^{D-2}}\right].
		\end{align}
		Using the mean value formula, we have
		\begin{eqnarray}\label{g-gbar}
		|(g-\bar{g})(u,r)|\leq \frac{8\pi \hat{x}^2}{(D-3)(D-2)^2}\frac{r^{D-4}}{(1+u)^{2k-3}(1+r+u)^{D-3}}.
		\end{eqnarray}
		Combining (\ref{estimate g}) and (\ref{g-gbar}), and using triangle inequality we get
		\begin{eqnarray}\label{estimate gbar}
		|\bar{g}(u,r)|\geq\frac{8\pi \hat{x}^2}{(D-3)(D-2)^2}\frac{1}{(1+u)^{2k-3}(1+r+u)},
		\end{eqnarray}
		and
		\begin{align}\label{integral gbar}
		\frac{(D-4)}{2}\frac{\hat{R}(\hat{\sigma})}{r^{D-3}}\int_{0}^{r}|\bar{g}|s^{D-4}\mathrm{d}s\leq \frac{C\hat{x}^2}{(1+u)^{2k-3}(1+r+u)}.
		\end{align}
		From (\ref{htilde local}), we represent
		\begin{align}\label{integral shbarg}
		\frac{8\pi}{r^{D-3}}\int_{0}^{r} gs^{D-2}|V(\tilde{h})|\mathrm{d}s\leq \frac{C\hat{x}^{p+1}}{(1+u)^{k^2-D}(1+r+u)^{D-3}},
		\end{align}
		for $p\in[k,\infty)$, where $p,k \in \mathbb{R}^+$. Thus, from (\ref{gtilde}) we have
		\begin{eqnarray}\label{g-gtilde k}
		\left|\frac{\hat{R}(\hat{\sigma})}{(D-2)}g-\frac{(D-2)}{2}\tilde{g}\right|&\leq&|g-\bar{g}|	+\frac{(D-4)}{2}\frac{\hat{R}(\hat{\sigma})}{r^{D-3}}\int_{0}^{r}|\bar{g}|s^{D-4}\mathrm{d}s\nonumber\\
		&&+\frac{8\pi}{r^{D-3}}\int_{0}^{r} gs^{D-2}|V(\tilde{h})|\mathrm{d}s\nonumber\\
		&\leq& \frac{C(\hat{x}^2+\hat{x}^{p+1})}{(1+u)^{2k-3}(1+r+u)}.
		\end{eqnarray}
		Combining (\ref{htilde local}) and (\ref{g-gtilde k}), we write the estimate
		\begin{eqnarray}\label{fk local}
		|f|\leq \frac{C(\hat{x}^3+\hat{x}^p+\hat{x}^{p+2})}{(1+u)^{\frac{(2k-D)k}{2}}(1+r+u)^{\frac{(D-2)k-2}{2}}}.
		\end{eqnarray}
		
		From the definition $\hat{d}=\|h(0,)\|_{\hat{X}_0}$, we obtain
		\begin{eqnarray}\label{h local}
		|h(0,r)|\leq \frac{C\hat{d}}{(1+r)^{k-1}}.
		\end{eqnarray}
		
		We assume that $h(u,r)$ is defined in $\mathcal{I}=[0,\delta]$ that contains $[0,u_0]$ such that $\delta\in \mathcal{I}$ but $\delta\notin [0,u_0]$. Hence, we estimate the exponential term of (\ref{Fcurl}) from $u\in[0,u_0]$ up to $u'=\delta>u_0$ such that
		\begin{align}\label{delta1}
		\int_{0}^{\delta} \left[\frac{1}{2r}\left|\frac{\hat{R}(\hat{\sigma})}{(D-2)}g-\frac{(D-2)}{2}\tilde{g}\right|+\frac{8\pi g r|V(\tilde{h})|}{(D-2)}\right]\mathrm{d}u\leq C(\hat{x}^2+\hat{x}^{p+1})\delta.
		\end{align}
		Combining (\ref{fk local}), (\ref{h local}), and (\ref{delta1}), we obtain the estimate for (\ref{Fcurl}) as follows
		\begin{eqnarray}\label{Fcurl local}
		|\mathcal{F}(\delta,r)|\leq\frac{\hat{C}_1(\hat{d}+\hat{x}^3+\hat{x}^p+\hat{x}^{p+2})\delta\exp\left[\hat{C}_2(\hat{x}^2+\hat{x}^{p+1})\delta\right]}{(1+r)^{k-1}},
		\end{eqnarray}
		for any $\hat{C}_1, \hat{C}_2 >0$ depends on $k$ and dimensions $D$.
		
		Now, we calculate
		\begin{align}\label{dgtilde dr}
		\left|\frac{\partial \tilde{g}}{\partial r}\right|\leq&\frac{\hat{R}(\hat{\sigma})}{(D-2)}\frac{|g-\bar{g}|}{r}+\frac{\hat{R}(\hat{\sigma})(D-4)(D-3)}{(D-2)r^{D-2}}\int_{0}^{r}|\bar{g}|s^{D-4}\mathrm{d}s+\frac{\hat{R}(\hat{\sigma})(D-4)}{(D-2)}\frac{|\bar{g}|}{r}\nonumber\\
		&+\frac{(D-3)16\pi}{(D-2)r^{D-2}}\int_{0}^{r} \left|g\right|s^{D-2}\left|V(\tilde{h})\right|\mathrm{d}s+\frac{16\pi |g|r\left|V(\tilde{h})\right|}{(D-2)}.
		\end{align}
		From (\ref{g-gbar}), we get
		\begin{align}\label{g-gbar r}
		\frac{|g-\bar{g}|}{r}\leq \frac{C\hat{x}^2}{(1+u)^{2k-3}(1+r+u)^2}.
		\end{align}
		Using (\ref{estimate gbar}), we obtain
		\begin{align}
		\frac{\hat{R}(\hat{\sigma})(D-4)(D-3)}{(D-2)r^{D-2}}\int_{0}^{r}|\bar{g}|s^{D-4}\mathrm{d}s\leq\frac{C\hat{x}^2}{(1+u)^{2k-3}(1+r+u)^2},
		\end{align}
		and
		\begin{align}
		\frac{(D-4)}{(D-2)}\hat{R}(\hat{\sigma})\frac{1}{r}|\bar{g}|\leq\frac{8\pi \hat{x}^2}{(D-3)(D-2)^2(1+u)^{2k-3}(1+r+u)^2}.
		\end{align}
		Then, using (\ref{htilde}) we have
		\begin{align}
		\frac{(D-3)16\pi}{(D-2)r^{D-2}}\int_{0}^{r} \left|g\right|s^{D-2}\left|V(\tilde{h})\right|\mathrm{d}s\leq\frac{C\hat{x}^{p+1}}{(1+u)^{k^2-D}(1+r+u)^{D-2}},
		\end{align}
		and
		\begin{align}
		\frac{16\pi gr}{(D-2)}|V(\tilde{h})|\leq \frac{C\hat{x}^{p+1}r}{(1+u)^{\frac{(2k-D)(k+1)}{2}}(1+r+u)^{\frac{(D-2)(k+1)}{2}}}.
		\end{align}
		Furthermore, we have
		\begin{align}\label{gtilde higher}
		\frac{(D-2)}{2}\left|\frac{\partial\tilde{g}}{\partial r}\right|\leq \frac{C(\hat{x}^2+\hat{x}^{k+1}+\hat{x}^{p+1})r}{(1+u)^{2k-3}(1+r+u)^3}.
		\end{align}
		From (\ref{estimate g}) we estimate
		\begin{eqnarray}\label{gr higher}
		\frac{\hat{R}(\hat{\sigma})}{(D-2)}\left|\frac{\partial g}{\partial r}\right|\leq\frac{C\hat{x}^2r}{(1+u)^{2k-D}(1+r+u)^{D}}.
		\end{eqnarray}
		Combination of (\ref{gtilde higher}) and (\ref{gr higher}), yields
		\begin{eqnarray}\label{B1}
		\left|\frac{1}{2r}\left(\frac{\hat{R}(\hat{\sigma})}{(D-2)}\frac{\partial g}{\partial r}-\frac{(D-2)}{2}\frac{\partial\tilde{g}}{\partial r}\right)\right|\leq\frac{C(\hat{x}^2+\hat{x}^{k+1}+\hat{x}^{p+1})}{(1+u)^{2k-D}(1+r+u)^3}.
		\end{eqnarray}
		On the other hand, from (\ref{htilde}) we obtain the relation $\left|\frac{\partial \tilde{h}}{\partial r}\right|=\frac{\left|h-\frac{(D-2)}{2}\tilde{h}\right|}{r}$ such that
		\begin{eqnarray}\label{B2}
		\left|\frac{8\pi g r}{(D-2)} \frac{\partial V(\tilde{h})}{\partial \tilde{h}} \frac{\partial \tilde{h}}{\partial r}\right|
		\leq\frac{C\hat{x}^{p+1}}{(1+u)^{\frac{(2k-D)(k+1)}{2}}(1+r+u)^{\frac{(D-2)(k+1)}{2}}}.
		\end{eqnarray}
		In the view of (\ref{B2}), (\ref{B1}), (\ref{g-gtilde k}), and (\ref{htilde}) we obtain
		\begin{align}\label{f1 local}
		|f_1|\leq&\frac{C(\hat{x}^2+\hat{x}^{k+1}+\hat{x}^{p+1}+\hat{x}^{p+3})}{(1+u)^{2k-D}(1+r+u)^3}|\mathcal{F}|+\frac{C(\hat{x}^3+\hat{x}^{k+2}+\hat{x}^{p+2}+\hat{x}^{p+4})}{(1+u)^{2k-3}(1+r+u)^{\frac{D+4}{2}}}\nonumber\\
		&+\frac{C(\hat{x}^3+\hat{x}^{p}+\hat{x}^{p+2})r}{(1+u)^{2k-3}(1+r+u)^{\frac{D+4}{2}}}.
		\end{align}
		
		Then, we calculate
		\begin{eqnarray}
		\left|\frac{\partial h}{\partial r}(0,r) \right|\leq \frac{C\hat{d}}{(1+r)^k}.
		\end{eqnarray}
		
		Using the similar calculation as the estimate (\ref{delta1}), we obtain
		\begin{align}\label{exp local}
		\int_{0}^{\delta}\left[\frac{1}{2}\left|\frac{\partial\tilde{g}}{\partial r}\right|+\frac{1}{2r}\left|\left(\frac{\hat{R}(\hat{\sigma})}{(D-2)}g-\frac{(D-2)}{2}\tilde{g}\right) \right|+\frac{8\pi g r|V(\tilde{h})|}{(D-2)}\right]\mathrm{d}u\leq C\left(\hat{x}^2 + \hat{x}^{k+1} + \hat{x}^{p+1} \right)\delta.
		\end{align}
		Combining (\ref{f1 local})-(\ref{exp local}) together with (\ref{Fcurl}), we write the estimate for (\ref{Gcurl}) as follows
		\begin{eqnarray}\label{Gcurl local}
		\left|\mathcal{G}(\delta,r)\right|
		&\leq&\frac{\hat{C}_3(\hat{d}+\hat{x}^3+\hat{x}^{k+2}+\hat{x}^p+\hat{x}^{p+2}+\hat{x}^{p+4})(1+\hat{x}^2+\hat{x}^{k+1}+\hat{x}^{p+1}+\hat{x}^{p+3})\delta}{(1+r)^k}\nonumber\\
		&&\times \exp\left[\hat{C}_4(\hat{x}^2+\hat{x}^{k+1}+\hat{x}^{p+1})\delta\right],
		\end{eqnarray}
		for any $\hat{C}_3, \hat{C}_4 >0$ depends on $k$ and dimensions $D$.
		
		From the definition of norm for the space $\hat{X}$, we obtain
		\begin{eqnarray}
		\|\mathcal{F}\|_{\hat{X}}&\leq&\hat{C}_5(\hat{d}+\hat{x}^3+\hat{x}^{k+2}+\hat{x}^p+\hat{x}^{p+2}+\hat{x}^{p+4})(1+\hat{x}^2+\hat{x}^{k+1}+\hat{x}^{p+1}+\hat{x}^{p+3})\delta\nonumber\\
		&&\times\exp\left[\hat{C}_6(\hat{x}^2+\hat{x}^{k+1}+\hat{x}^{p+1})\delta\right]\;.
		\end{eqnarray}
		for any $\hat{C}_5, \hat{C}_6 >0$ depends on $k$ and dimensions $D$.
		
		Let us introduce a function
		\begin{eqnarray}
		\mathcal{L}_1(\hat{x})=\frac{\hat{x}\exp\left[-\hat{C}_6(\hat{x}^2+\hat{x}^{k+1}+\hat{x}^{p+1})\delta\right]}{\hat{C}_5(1+\hat{x}^2+\hat{x}^{k+1}+\hat{x}^{p+1}+\hat{x}^{p+3})\delta}-\left(\hat{x}^3+\hat{x}^{k+2}+\hat{x}^p+\hat{x}^{p+2}+\hat{x}^{p+4}\right),
		\end{eqnarray}
		satisties $\mathcal{L}_1(0)=0$, $\mathcal{L}_1'(0)>0$, and $\mathcal{L}_1(\hat{x})\rightarrow-\infty$ as $\hat{x}\rightarrow\infty$. Furthermore, there exists $\hat{x}_0\in(0,\hat{x}_1)$ such that $\mathcal{L}_1(\hat{x})$ is monotonically increasing on $[0,\hat{x}_0]$. For every $\hat{x}\in(0,\hat{x}_0)$, we obtain that $\|\mathcal{F}\|_{\hat{X}}\leq \hat{x}$ such that the map $\mathcal{F}$ defined in (\ref{Fcurl}) is contained in the closed ball of radius $\hat{x}$ in the space $\hat{X}$ if $\hat{d}<\mathcal{L}_1(\hat{x})$. 
		This is the end of the proof.
	\end{proof}
	\begin{lemma}
	Let $D\geq 4$. There exists $\delta:=\delta(\hat{x})>0$, such that if $u_0<\delta$, the map  $\mathcal{F}$ defined in (\ref{Fcurl}) contracts in $\hat{Y}$.
	\end{lemma}
	\begin{proof}
		Suppose that equation (\ref{Dh higher}) has two different solutions $h_1, h_2 \in \hat{X}$, that are defined in $\mathcal{I}=[0,\delta]$ and contains $[0,u_0]$ such that $\delta\in \mathcal{I}_2$ but $\delta\notin [0,u_0]$. We assume
		\begin{align}
		\max\{\|h_1\|_{\hat{X}},\|h_2\|_{\hat{X}}\}<\hat{x}.
		\end{align}
		
		Let us consider the notion $\|h_1-h_2\|_{\hat{Y}}=\hat{y}$. Using (\ref{htilde}) we get
		\begin{eqnarray}\label{estimate h1-h2}
		|\tilde{h}_1 - \tilde{h}_2|&\leq& r^{-\frac{(D-2)}{2}}\int_{0}^{r}|h_1-h_2|s^{\frac{D-4}{2}}\mathrm{d}s\nonumber\\
		&\leq&r^{-\frac{(D-2)}{2}}\int_{0}^{r}\frac{\|h_1-h_2\|_{\hat{Y}}}{(1+s+u)^{k-1}}s^{\frac{D-4}{2}}\mathrm{d}s\nonumber\\
		&=&\frac{2\hat{y}}{(2k-D)}\frac{1}{(1+u)^{\frac{(2k-D)}{2}}(1+r+u)^{\frac{(D-2)}{2}}}.
		\end{eqnarray}
		We estimate
		\begin{eqnarray}
		|h_1-h_2-(\tilde{h}_1-\tilde{h}_2)|\leq |h_1-h_2|+|\tilde{h}_1-\tilde{h}_2|\leq\frac{C\hat{y}}{(1+u)^{\frac{(2k-D)}{2}}(1+r+u)^{\frac{(D-2)}{2}}},
		\end{eqnarray}
		and
		\begin{eqnarray}
		\left||h_1-\tilde{h}_1|^2-|h_2-\tilde{h}_2|^2\right|&\leq& \left|(h_1-h_2)-(\tilde{h}_1-\tilde{h}_2) \right|\left(|h_1-\tilde{h}_1|+|h_2-\tilde{h}_2|\right)\nonumber\\
		&\leq& \frac{C\hat{x}\hat{y}}{(1+u)^{2k-D}(1+r+u)^{D-2}}.
		\end{eqnarray}
		Thus,
		\begin{eqnarray}\label{g1g2}
		|g_1-g_2|&\leq&\frac{8\pi}{(D-2)} \int_{r}^{\infty}\frac{1}{s}\left||h_1-\tilde{h}_1|^2-|h_2-\tilde{h}_2|^2\right|\mathrm{d}s\nonumber\\
		&\leq& \frac{C\hat{x}\hat{y}}{(1+u)^{(2k-D)}(1+r+u)^{(D-2)}},
		\end{eqnarray}
		such that
		\begin{align}\label{estimate g1g2}
		|\bar{g}_1-\bar{g}_2|\leq\frac{1}{r}\int_{0}^{r}|g_1-g_2|\mathrm{d}s\leq\frac{C\hat{x}\hat{y}}{(1+u)^{2k-3}(1+r+u)},
		\end{align}
		and
		\begin{align}\label{estimate g1g2 2}
		\frac{(D-4)}{(D-2)}\frac{\hat{R}(\hat{\sigma})}{r^{D-3}}\int_{0}^{r}\left|\bar{g}_1-\bar{g}_2\right|s^{D-4}\mathrm{d}s\leq\frac{C\hat{x}\hat{y}}{(1+u)^{2k-3}(1+r+u)}.
		\end{align}
		Using the notion
		\begin{eqnarray}
		\tilde{h}_1^{p+1} - \tilde{h}_2^{p+1}=(\tilde{h}_1-\tilde{h}_2)\int_{0}^{1}\left(t\tilde{h}_1+(1-t)\tilde{h}_2\right)\left|t\tilde{h}_1+(1-t)\tilde{h}_2\right|^{p-1}\mathrm{d}t,
		\end{eqnarray}
		we calculate
		\begin{eqnarray}\label{hp1hp2}
		\left|\tilde{h}_1^{p+1}-\tilde{h}_2^{p+1}\right|&\leq& |\tilde{h}_1-\tilde{h}_2|\left(|\tilde{h}_1|+|\tilde{h}_2|\right)^p\nonumber\\
		&\leq& \frac{2^{2p+1}\hat{x}^p\hat{y}}{(2k-D)^{k+1}(1+u)^{\frac{(2k-D)(k+1)}{2}}(1+r+u)^{\frac{(D-2)(k+1)}{2}}}.
		\end{eqnarray}
		From (\ref{hp1hp2}) we obtain
		\begin{align}\label{estimate hp1hp2}
		\frac{16\pi  }{(D-2)}\frac{1}{r^{D-3}}\int_{0}^{r} g_1s^{D-2}\left|V(\tilde{h}_1)-V(\tilde{h}_2)\right|\mathrm{d}s\leq \frac{C\hat{x}^p\hat{y}}{(1+u)^{k^2-D}(1+r+u)^{D-3}}.
		\end{align}
		Combination of (\ref{estimate g1g2}), (\ref{estimate g1g2 2}), and (\ref{estimate hp1hp2}) yields
		\begin{align}
		\left|\tilde{g}_1-\tilde{g}_2\right|\leq& \frac{\hat{R}(\hat{\sigma})}{(D-2)}\left|\bar{g}_1-\bar{g}_2\right|+\frac{(D-4)}{(D-2)}\frac{\hat{R}(\hat{\sigma})}{r^{D-3}}\int_{0}^{r}\left|\bar{g}_1-\bar{g}_2\right|s^{D-4}\mathrm{d}s\nonumber\\
		&+ \frac{16\pi}{(D-2)}\frac{1}{r^{D-3}}\int_{0}^{r} g_1s^{D-2}\left|V(\tilde{h}_1)-V(\tilde{h}_2)\right|\mathrm{d}s\nonumber\\
		\leq&\frac{C(\hat{x}+\hat{x}^p)\hat{y}}{(1+u)^{2k-3}(1+r+u)}.
		\end{align}
		Then, we estimate
		\begin{eqnarray}
		\left|\frac{\hat{R}(\hat{\sigma})}{(D-2)}(g_1-g_2)-\frac{\hat{R}(\hat{\sigma})}{2}(\bar{g}_1-\bar{g}_2)\right|\leq\frac{C\hat{x}\hat{y}r}{(1+u)^{2k-3}(1+r+u)^{D-2}},
		\end{eqnarray}
		such that
		\begin{align}
		&\frac{1}{2r}\left|\frac{\hat{R}(\hat{\sigma})}{(D-2)}(g_1-g_2)-\frac{(D-2)}{2}(\tilde{g}_1-\tilde{g}_2)\right|\leq\frac{1}{2r}\left[	\left|\frac{\hat{R}(\hat{\sigma})}{(D-2)}(g_1-g_2)-\frac{\hat{R}(\hat{\sigma})}{2}(\bar{g}_1-\bar{g}_2)\right|\right.\nonumber\\
		&\left.+\frac{(D-4)}{2}\frac{\hat{R}(\hat{\sigma})}{r^{D-3}}\int_{0}^{r}\left|\bar{g}_1-\bar{g}_2\right|s^{D-4}\mathrm{d}s+\frac{8\pi}{r^{D-3}}\int_{0}^{r} g_1s^{D-2}\left|V(\tilde{h}_1)-V(\tilde{h}_2)\right|\mathrm{d}s\right]\nonumber\\
		&\leq\frac{C(\hat{x}+\hat{x}^p)\hat{y}}{(1+u)^{2k-3}(1+r+u)^{D-2}}.
		\end{align}
		Combining the above estimates together with (\ref{Gcurl local}), (\ref{Fcurl local}), and (\ref{htilde local}), we write
		\begin{equation}
		|\tilde{\psi}|\leq\frac{C\hat{y}\left[\hat{\alpha}(\hat{x})+\hat{\beta}(\hat{x})+\hat{\gamma}(\hat{x})+\hat{\sigma}(\hat{x})+\hat{x}^2+\hat{x}^{p-1}+\hat{x}^{p+1}+\hat{x}^{p+3}\right]}{(1+u)^{\frac{3(2k-D)}{2}}(1+r+u)^{D-2}}\;,
		\end{equation}
		where
		\begin{align}
		\hat{\alpha}(\hat{x})=&C(\hat{x}+\hat{x}^p)\left(\hat{d}+\hat{x}^3+\hat{x}^{k+2}+\hat{x}^p+\hat{x}^{p+2}+\hat{x}^{p+4}\right)(1+\hat{x}^2+\hat{x}^{k+1}+\hat{x}^{p+1}+\hat{x}^{p+3})\delta\nonumber\\
		&\times\exp\left[C(\hat{x}^2+\hat{x}^{k+1}+\hat{x}^{p+1})\delta\right],\\
		\hat{\beta}(\hat{x})=&C(\hat{x}+\hat{x}^p)(\hat{d}+\hat{x}^3+\hat{x}^p+\hat{x}^{p+2})\delta\exp\left[C(\hat{x}^2+\hat{x}^{p+1})\delta\right],\\
		\hat{\gamma}(\hat{x})=&C \hat{x}^p(\hat{d}+\hat{x}^3+\hat{x}^p+\hat{x}^{p+2})\delta\exp[C(\hat{x}^2+\hat{x}^{p+1})\delta],\\
		\hat{\sigma}(\hat{x})=&C\hat{x}^{p+2}(\hat{d}+\hat{x}^3+\hat{x}^p+\hat{x}^{p+2})\delta\exp[C(\hat{x}^2+\hat{x}^{p+1})\delta].
		\end{align}
		
		We write the estimate for the exponential term of (\ref{Theta}) from $u\in[0,u_0]$ up to $u'=\delta>u_0$ such that
		\begin{eqnarray}
		\int_{0}^{\delta} \left[\frac{1}{2r}\left|\frac{\hat{R}(\hat{\sigma})}{(D-2)}g_1-\frac{(D-2)}{2}\tilde{g}_1\right|+\frac{8\pi r g_2|V(\tilde{h}_1)|}{(D-2)}\right]\mathrm{d}u'\leq C(\hat{x}^2+\hat{x}^{p+1})\delta.
		\end{eqnarray}
		Therefore, the estimate for (\ref{Theta}) yields
		\begin{eqnarray}
		\left|\Theta(\delta,r)\right|\leq \hat{C}_7\frac{\hat{y} \hat{M}(\hat{x})\delta\exp[\hat{C}_8(\hat{x}^2+\hat{x}^{p+1})\delta]}{(1+r)^{k-1}},
		\end{eqnarray}
		for any $\hat{C}_7, \hat{C}_8 >0$ depends on $k$ and dimensions $D$,	and we have denoted
		\begin{eqnarray}
		\hat{M}(\hat{x}) = \hat{\alpha}(\hat{x})+\hat{\beta}(\hat{x})+\hat{\gamma}(\hat{x})+\hat{\sigma}(\hat{x})+\hat{x}^2+\hat{x}^{p-1}+\hat{x}^{p+1}+\hat{x}^{p+3}.
		\end{eqnarray}
		
		We represent the Lipschitz condition for (\ref{Theta}) in the space $\hat{Y}$ as follows
		\begin{equation}
		\|\Theta\|_{\hat{Y}} \leq\mathcal{L}_2\hat{y}
		\end{equation}
		with
		\begin{equation}
		\mathcal{L}_2=\hat{C}_7M(\hat{x})\delta\exp[\hat{C}_8(\hat{x}^2+\hat{x}^{p+1})\delta].\:
		\end{equation}
		We deduce that $\mathcal{L}_2(0)=0$ and $\mathcal{L}_2'(0)>0$. 
		Also, $\mathcal{L}_2$ is monotonically increasing on $\hat{x}_2\in \mathbb{R}^+$. There exists $\hat{x}_3\in \mathbb{R}^+$ such that $\mathcal{L}_2(x)<1$ for all $\hat{x}$ in $(0,\hat{x}_3]$. 
		Hence, the mapping $h\mapsto\mathcal{F}(h)$ contracts in $\hat{Y}$ for $\|h\|_{\hat{X}}\leq \hat{x}_3$. The proof is finished.
	\end{proof}
	Since $\hat{Y}$ containing $\hat{X}$, and $h\mapsto\mathcal{F}(h)$ contracts in $\hat{Y}$, then $h\mapsto\mathcal{F}(h)$ also contracts in $\hat{X}$. As a consequance, there exists a unique fixed point $h\in \hat{X}$ such that $\mathcal{F}(h)=h$. From (\ref{Dh higher}), we obtain that $\left|\frac{\partial h}{\partial u}\right|$ bounded, namely
	\begin{eqnarray}\label{dhdu hat}
	\left|\frac{\partial h}{\partial u}\right|&\leq& \frac{\tilde{g}}{2}\left|\frac{\partial h}{\partial r}\right| + \frac{1}{2r}\left|\frac{\hat{R}(\hat{\sigma})}{(D-2)}g-\frac{(D-2)}{2}\tilde{g}\right|\left|h\right|+\frac{(D-2)}{4r}\left|\frac{\hat{R}(\hat{\sigma})}{(D-2)}g-\frac{(D-2)}{2}\tilde{g}\right|\left|\tilde{h}\right|\nonumber\\
	&&+\frac{8\pi gr}{(D-2)}\left|h\right||V(\tilde{h})|+4\pi gr|\tilde{h}||V(\tilde{h})|+\frac{gr}{2}\left|\frac{\partial V(\tilde{h})}{\partial \tilde{h}}\right|\nonumber\\
	&\leq&\frac{C(\hat{K}_1(x)+\hat{K}_2(x)+\hat{K}_3+\hat{x}^3+\hat{x}^p+\hat{x}^{p+2})}{(1+r)^{\frac{3(D-2)}{2}-1}},
	\end{eqnarray}
	in $[0,u_0]\times[0,\infty)$
	with
	\begin{align}
	\hat{K}_1(\hat{x})=&C(\hat{d}+\hat{x}^3+\hat{x}^{k+2}+\hat{x}^p+\hat{x}^{p+2}+\hat{x}^{p+4})(1+\hat{x}^2+\hat{x}^{k+1}+\hat{x}^{p+1}+\hat{x}^{p+3})\delta\nonumber\\
	&\times\exp\left[C(\hat{x}^2+\hat{x}^{k+1}+\hat{x}^{p+1})\delta\right]\\
	\hat{K}_2(\hat{x})=&C(\hat{x}^2+\hat{x}^{p+1})(\hat{d}+\hat{x}^3+\hat{x}^p+\hat{x}^{p+2})\delta\exp\left[C(\hat{x}^2+\hat{x}^{p+1})\delta\right]\\
	\hat{K}_3(\hat{x})=&C\hat{x}^{p+1}(\hat{d}+\hat{x}^3+\hat{x}^p+\hat{x}^{p+2})\delta\exp\left[C(\hat{x}^2+\hat{x}^{p+1})\delta\right].
	\end{align}
	Estimate (\ref{dhdu hat}) ensures that $h$ is a classical solution of the equation (\ref{Dh higher}).  
	
	Finally, we state the theorem of local existence as follows
	\begin{theorem}
		Setting $D\geq 4$. Let $\hat{X}$ and $\hat{Y}$ be the function spaces defined by (\ref{space Xhat}) and (\ref{space Yhat}) respectively. For a given an initial data $h(0,r)\in C^1[0,\infty)$ and the positive constant $(D-2)$-spatial Ricci scalar of compact manifold $\hat{R}(\hat{\sigma})$ such that $h(0,r)=O(r^{-{(k-1)}})$ and $\frac{\partial h}{\partial r}(0,r)=O(r^{-k})$ as $r\rightarrow \infty$, there exist a $u_0>0$ and a classical solution $h(u,r)\in C^1 ([0,u_0]\times [0,\infty))$ of the equation (\ref{Dh higher}) for $k>\frac{D}{2}$ and $p\in[k,\infty)$, such that $h(u,r)=O(r^{-{(k-1)}})$ and $\frac{\partial h}{\partial r}(u,r)=O(r^{-k})$ as $r\rightarrow \infty$ at each $u\in[0,u_0]$.
	\end{theorem}
	\subsection{Global Existence}
	This section is devoted to prove the existence of global solution of the equation (\ref{Dh higher}). We employ the similar method with the proof of local existence by taking $u_0\rightarrow\infty$.
	
	Given the spaces of function
	\begin{eqnarray}
	X & = & \{h\in C^1([0,\infty) \times [0,\infty) )\; | \;\|h\|_X < \infty\}\:, \label{space X}\\
	X_0 & = & \{h\in C^1([0,\infty) )\; | \;\|h\|_{X_0} < \infty\}\:, \\
	Y & = & \{h\in C^1([0,\infty) \times [0,\infty) )\; | \; h(0,r) = h_0(r), \|h\|_Y < \infty\}\label{space Y}\:,
	\end{eqnarray}
	in which $Y$ containing $X$, equipped by the norm
	\begin{eqnarray}
	\|h\|_X & = & \sup_{u\geq 0} \sup_{r\geq 0}\left\{ (1+r+u)^{k-1}|h(u,r)| + (1+r+u)^{k}\left|\frac{\partial h}{\partial r}(u,r)\right|  \right\}\:, \label{hx}\\
	\|h\|_{X_0} & = & \sup_{r\geq 0}\left\{ (1+r)^{k-1}|h(r)| + (1+r)^{k}\left|\frac{\partial h}{\partial r}(r)\right|  \right\}\:, \\
	\|h\|_Y & = & \sup_{u\geq 0} \sup_{r\geq 0}\left\{ (1+r+u)^{k-1}|h(u,r)|\right\}\:\label{norm Y},
	\end{eqnarray}
	with $k>\frac{D}{2}$, where $k\in \mathbb{R}^+$, and $D\geq 4$. Here, the space function $X, X_0,$ and $Y$ are the Banach space. We also denote $\|h\|_X=x$, $\|h_0\|_{X_0}=d$, and $\|h_1-h_2\|_Y=y$.
	
	\begin{lemma} \label{Lemma 1}
		Let us denote $\|h\|_X=x$ and $\|h(0,.)\|_{X_0}=d$. For $D\geq 4$, setting $k>\frac{D}{2}$ and $p\in[k,\infty)$. Given the initial data $h(0,r)\in C^1[0,\infty)$ and the positive constant $(D-2)$-spatial Ricci scalar of compact manifold $\hat{R}(\hat{\sigma})$ such that $h(0,r)=O(r^{-{(k-1)}})$ and $\frac{\partial h}{\partial r}(0,r)=O(r^{-k})$ as $r\rightarrow \infty$.  Then, the solution of (\ref{Dh higher}) has the decay properties:
		\begin{align} \label{decay1}
		|h(u,r)|\leq \frac{C(d+x^3+x^p+x^{p+2})\exp\left[C(x^2+x^{p+1})\right]}{(1+r+u)^{k-1}},
		\end{align}
		and
		\begin{align} \label{decay2}
		\left|\frac{\partial h}{\partial r}(u,r) \right|\leq& \frac{C(d+x^3+x^{k+2}+x^p+x^{p+2}+x^{p+4})(1+x^2+x^{k+1}+x^{p+1}+x^{p+3})}{(1+r+u)^k}\nonumber\\
		&\times \exp\left[C(x^2+x^{k+1}+x^{p+1})\right].
		\end{align}
		for any $C:=C(k,D)>0$ depends on $k$ and dimensions $D$.
	\end{lemma}
	\begin{proof}	
		First, we calculate
		\begin{eqnarray}\label{estimate htilde}
		|\tilde{h}|\leq \frac{2x}{(2k-D)}\frac{1}{(1+u)^{\frac{(2k-D)}{2}}(1+r+u)^{\frac{(D-2)}{2}}}.
		\end{eqnarray}
		As previously, we obtain the estimate for (\ref{f prisma}) as follows (see Appendix \ref{Appendix1})
		\begin{eqnarray}\label{fk}
		|f|\leq \frac{C(x^3+x^p+x^{p+2})}{(1+u)^{\frac{(2k-D)k}{2}}(1+r+u)^{\frac{(D-2)k-2}{2}}}.
		\end{eqnarray}
		
		From (\ref{gtilde}), we get
		\begin{eqnarray}
		|\tilde{g}(u,0)|\geq\frac{2^4\pi x^2\hat{R}(\hat{\sigma})}{(D-3)(D-2)^3}+\frac{2^{p+7}\pi K_0x^{p+1}}{(D-2)\left[(D-2)k-D\right](2k-D)^{k+1}}.
		\end{eqnarray}
		Let us denote $x_1$ as the positive solution of
		\begin{eqnarray}
		\frac{2^4\pi x^2\hat{R}(\hat{\sigma})}{(D-3)(D-2)^3}+\frac{2^{p+7}\pi K_0x^{p+1}}{(D-2)\left[(D-2)k-D\right](2k-D)^{k+1}}=0.
		\end{eqnarray}
		Then, we define a new function
		\begin{eqnarray}
		\kappa:=\kappa(x)=\frac{2^4\pi x^2\hat{R}(\hat{\sigma})}{(D-3)(D-2)^3}+\frac{2^{p+7}\pi K_0x^{p+1}}{(D-2)\left[(D-2)k-D\right](2k-D)^{k+1}}.
		\end{eqnarray}
		for all $x\in[0,x_1)$. Using the definition of characteristics we have
		\begin{eqnarray}
		r(u)=r_1 + \frac{1}{2}\int_{u}^{u_1}\tilde{g}(u',r(u'))\mathrm{d}u'\geq r_1 +\frac{1}{2}\kappa(u_1-u),
		\end{eqnarray}
		and
		\begin{eqnarray}\label{ru}
		1+r(u)+u\geq 1 + \frac{u}{2} + r_1 + \frac{\kappa}{2}(u_1 - u)\geq \kappa\left(1+\frac{u_1}{2}+r_1\right)\geq\frac{\kappa}{2}(1+r_1+u_1).
		\end{eqnarray}
		Thus, we obtain
		\begin{eqnarray}\label{h0r0k}
		|h(0,r_0)|\leq \frac{\|h(0,)\|_{X_0}}{(1+r_0)^{k-1}}\leq\frac{2^{k-1}d}{\kappa^{k-1}(1+r_1+u_1)^{k-1}}.
		\end{eqnarray}
		
		Now, we introduce the integral formula
		\begin{eqnarray}\label{integral}
		\int_{0}^{u_1}\left[\frac{r^s}{(1+u)^t(1+r+u)^q}\right]_\chi\mathrm{d}u&\leq& \int_{0}^{u_1}\left[\frac{1}{(1+u)^t(1+r+u)^{q-s}}\right]_\chi\mathrm{d}u\nonumber\\
		&\leq& \frac{1}{\kappa^m(1+r_1+u_1)^m}\int_{0}^{\infty}\frac{\mathrm{d}u}{(1+u)^{q-s+t-m}}\nonumber\\
		&=&\frac{2^m}{(q-s+t-m-1)\kappa^m(1+r_1+u_1)^m} 
		\end{eqnarray}
		where $q-s+t-m>1$, with $q, s, t, m\in \mathbb{R}$.
		
		Setting $q=D-2$, $s=D-4$, $t=2k-3$, and $m=0$, we obtain
		\begin{align}
		\int_{0}^{u_1}\left|\frac{1}{2r}\left[\frac{\hat{R}(\hat{\sigma})}{(D-2)}g-\frac{(D-2)}{2}\tilde{g}\right]\right|_\chi\mathrm{d}u\leq& C (x^2 + x^{p+1})\int_{0}^{u_1}\left[\frac{r^{D-4}}{(1+u)^{2k-3}(1+r+u)^{D-2}}\right]_\chi\mathrm{d}u\nonumber\\
		\leq& C(x^2+x^{p+1})\int_{0}^{\infty}\frac{1}{(1+u)^{2k-1}}\mathrm{d}u\leq C(x^2+x^{p+1}).
		\end{align}
		On the other hand, setting $q=\frac{(D-2)(k+1)-2}{2}$, $s=0$, $t=\frac{(2k-D)(k+1)}{2}$, and $m=0$, yields
		\begin{align}
		\int_{0}^{u_1} \left|\frac{8\pi g rV(\tilde{h})}{(D-2)}\right|_\chi~\mathrm{d}u\leq&
		Cx^{p+1}\int_{0}^{u_1}\left[\frac{1}{(1+u)^{\frac{(2k-D)(k+1)}{2}}(1+r+u)^{\frac{(D-2)(k+1)-2}{2}}}\right]_\chi\mathrm{d}u\nonumber\\
		\leq& Cx^{p+1}\int_{0}^{\infty}\frac{1}{(1+u)^{4(k-1)-1}}\mathrm{d}u
		\leq Cx^{p+1}.
		\end{align}
		Hence, we estimate the exponential term of (\ref{Fcurl}) as follows
		\begin{align}\label{expk}
		\int_{0}^{u_1} \left|\frac{1}{2r}\left[\frac{\hat{R}(\hat{\sigma})}{(D-2)}g-\frac{(D-2)}{2}\tilde{g}\right]+\frac{8\pi g rV(\tilde{h})}{(D-2)}\right|_\chi\mathrm{d}u\leq C(x^2+x^{p+1}).
		\end{align}
		Combining (\ref{fk}), (\ref{h0r0k}), and (\ref{expk}), we obtain the estimate for (\ref{Fcurl}) as follows
		\begin{eqnarray}\label{estimate Fcurl}
		|\mathcal{F}(u_1,r_1)|\leq\frac{C(d+x^3+x^p+x^{p+2})\exp\left[C(x^2+x^{p+1})\right]}{\kappa^{k-1}(1+r_1+u_1)^{k-1}}.
		\end{eqnarray}
		
		As previously, we calculate (see Appendix \ref{Appendix2})
		\begin{align}\label{f1 higher}
		|f_1|\leq&\frac{C(x^2+x^{k+1}+x^{p+1}+x^{p+3})}{(1+u)^{2k-D}(1+r+u)^3}|\mathcal{F}|+\frac{C(x^3+x^{k+2}+x^{p+2}+x^{p+4})}{(1+u)^{2k-3}(1+r+u)^{\frac{D+4}{2}}}\nonumber\\
		&+\frac{C(x^3+x^{p}+x^{p+2})r}{(1+u)^{2k-3}(1+r+u)^{\frac{D+4}{2}}}.
		\end{align}
		Then, from (\ref{ru}) we obtain
		\begin{eqnarray}\label{h0r02k}
		\left|\frac{\partial h}{\partial r}(0,r_0) \right|\leq \frac{\|h_0\|_{X_0}}{(1+r_0)^k}\leq\frac{d}{(1+r_1+\frac{1}{2}\kappa u_1)^k}\leq \frac{2^kd}{\kappa^k(1+r_1+u_1)^k}.
		\end{eqnarray}
		Again, we use the integral formula (\ref{integral}), to obtain
		\begin{align}\label{exp2}
		\int_{0}^{u_1}\left[\frac{1}{2}\left|\frac{\partial\tilde{g}}{\partial r}\right|+\frac{1}{2r}\left|\left(\frac{\hat{R}(\hat{\sigma})}{(D-2)}g-\frac{(D-2)}{2}\tilde{g}\right) \right|+\frac{8\pi g r|V(\tilde{h})|}{(D-2)}\right]_\chi\mathrm{d}u\leq C\left(x^2 + x^{k+1} + x^{p+1} \right).
		\end{align}
		Combining (\ref{f1 higher})-(\ref{exp2}) together with (\ref{Fcurl}), we write the estimate for (\ref{Gcurl}) as follows
		\begin{eqnarray}\label{estimate Gcurl}
		\left|\mathcal{G}(u_1,r_1)\right|
		&\leq&\frac{C(d+x^3+x^{k+2}+x^p+x^{p+2}+x^{p+4})(1+x^2+x^{k+1}+x^{p+1}+x^{p+3})}{\kappa^k(1+r_1+u_1)^k}\nonumber\\
		&&\times \exp\left[C(x^2+x^{k+1}+x^{p+1})\right].
		\end{eqnarray}
		The proof is finished.
	\end{proof}

	\begin{lemma}\label{Lemma Banach}
	Let us denote $B(0,x)=\left\{f\in X| \|f\|_X\leq x \right\}$ as the close ball with radius $x$ in $X$. For a suitable $x$ and $d=d(x)$, the mapping $\mathcal{F}(.)$ satisfies the arguments as follows:
	\begin{enumerate}
		\item $\mathcal{F}: B(0,x)\rightarrow B(0,x)$, and
		\item there exists $\tilde{\Lambda}=\tilde{\Lambda}(x)\in[0,1)$, such that
		\begin{equation}
		\|\mathcal{F}(h_1)-\mathcal{F}(h_2)\|_Y \leq \tilde{\Lambda}\|h_1-h_2\|_Y.
		\end{equation}
		Thus, $\mathcal{F}$ contracts in $Y$.
	\end{enumerate}
	\end{lemma}
	\begin{proof}
	From the definition of norm for the space $X$, we have
	\begin{eqnarray}\label{Estimate F}
	\|\mathcal{F}\|_X&\leq&\frac{C_1}{\kappa^k}(d+x^3+x^{k+2}+x^p+x^{p+2}+x^{p+4})(1+x^2+x^{k+1}+x^{p+1}+x^{p+3})\nonumber\\
	&&\times\exp\left[C_2(x^2+x^{k+1}+x^{p+1})\right]\;.
	\end{eqnarray}
	Let us define a function,
	\begin{eqnarray}
	\tilde{\Lambda}_1(x)=\frac{x\kappa^k\exp\left[-C_2(x^2+x^{k+1}+x^{p+1})\right]}{C_1(1+x^2+x^{k+1}+x^{p+1}+x^{p+3})}-\left(x^3+x^{k+2}+x^p+x^{p+2}+x^{p+4}\right).
	\end{eqnarray}
	
	We obtain that $\tilde{\Lambda}_1(0)=0$, $\tilde{\Lambda}_1'(0)>0$, and $\tilde{\Lambda}_1(x)\rightarrow-\infty$ as $x\rightarrow\infty$. Furthermore, there exists $x_0\in(0,x_1)$ such that $\tilde{\Lambda}_1(x)$ is monotonically increasing on $[0,x_0]$. For every $x\in(0,x_0)$, we conclude that $\|\mathcal{F}\|_X\leq x$, and $\mathcal{F}:B(0,x)\rightarrow B(0,x)$ if $d<\tilde{\Lambda}_1(x)$. This proves the first argument. 
	
	We prove the second argument by calculating the estimate (see Appendix \ref{Appendix3})
	\begin{equation}\label{estimate psi2}
	|\tilde{\psi}|\leq\frac{Cy\left[\alpha(x)+\beta(x)+\gamma(x)+\sigma(x)+x^2+x^{p-1}+x^{p+1}+x^{p+3}\right]}{(1+u)^{\frac{3(2k-D)}{2}}(1+r+u)^{D-2}}\;,
	\end{equation}
	where
	\begin{align}
	\alpha(x)=&C(x+x^p)\left(d+x^3+x^{k+2}+x^p+x^{p+2}+x^{p+4}\right)(1+x^2+x^{k+1}+x^{p+1}+x^{p+3})\label{alpha}\nonumber\\
	&\times\exp\left[C(x^2+x^{k+1}+x^{p+1})\right],\\
	\beta(x)=&C(x+x^p)(d+x^3+x^p+x^{p+2})\exp\left[C(x^2+x^{p+1})\right],\\
	\gamma(x)=&C x^p(d+x^3+x^p+x^{p+2})\exp[C(x^2+x^{p+1})],\\
	\sigma(x)=&Cx^{p+2}(d+x^3+x^p+x^{p+2})\exp[C(x^2+x^{p+1})]\label{sigma}.
	\end{align}
	
	As previously, using integral formula (\ref{integral}) we have
	\begin{eqnarray}
	\int_{u}^{u_1} \left[\frac{1}{2r}\left|\frac{\hat{R}(\hat{\sigma})}{(D-2)}g_1-\frac{(D-2)}{2}\tilde{g}_1\right|+\frac{8\pi r g_2|V(\tilde{h}_1)|}{(D-2)}\right]_{\chi_1}\mathrm{d}u'\leq C(x^2+x^{p+1}).
	\end{eqnarray}
	
	Suppose that equation (\ref{Dh higher}) has two different solutions $h_1, h_2 \in X$. We assume
	\begin{align}
	\max\{\|h_1\|_{X},\|h_2\|_{X}\}<x.
	\end{align}
	Then, we write the estimate for (\ref{Theta}) as follows
	\begin{eqnarray}\label{Estimate Theta}
	\left|\Theta(u_1,r_1)\right|\leq \int_{0}^{u_1}\exp[C(x^2+x^{p+1})]\frac{CyM(x)}{(1+u)^{\frac{3(2k-D)}{2}}(1+r+u)^{D-2}}\mathrm{d}u \:,
	\end{eqnarray}
	where
	\begin{eqnarray}
	M(x) = \alpha(x)+\beta(x)+\gamma(x)+\sigma(x)+x^2+x^{p-1}+x^{p+1}+x^{p+3}.
	\end{eqnarray}
	Finally, we get the estimate
	\begin{equation}\label{estimate Theta}
	\left|\Theta(u_1,r_1)\right|\leq C_3\frac{yM(x)\exp[C_4(x^2+x^{p+1})]}{\kappa^{k-1}(1+r_1+u_1)^{k-1}}.
	\end{equation}
	
	From the definition of norm for the space $Y$, we have
	\begin{equation}\label{Lipschitz condition}
	\|\Theta\|_Y \leq\tilde{\Lambda}_2y
	\end{equation}
	with
	\begin{equation}\label{Lipschitz constant}
	\tilde{\Lambda}_2=\frac{C_3}{\kappa^{k-1}}M(x)\exp[C_4(x^2+x^{p+1})].\:
	\end{equation}
	We obtain that $\tilde{\Lambda}_2(0)=0$ and $\tilde{\Lambda}'_2(0)>0$. 
	Furthermore, $\tilde{\Lambda}_2$ is monotonically increasing on $\tilde{x}_0\in \mathbb{R}^+$. There exists $x_2\in \mathbb{R}^+$ such that $\tilde{\Lambda}_2(x)<1$ for all $x$ in $(0,x_2]$. 
	Hence, the mapping $h\mapsto\mathcal{F}(h)$ is contraction in $Y$ for $\|h\|_X\leq x_2$. This proves the second argument. 
	\end{proof}
	
	Given that $Y$ containing $X$, since $h\mapsto\mathcal{F}(h)$ is contraction mapping in $Y$, it is guaranteed that $h\mapsto\mathcal{F}(h)$ is also contraction mapping in $X$. Using Banach fixed theorem, there exists a unique fixed point $h\in X$ such that $\mathcal{F}(h)=h$.
	
	In the last part of this section, we show that $h$ is a classical solution of the equation (\ref{Dh higher}). For this purpose, we will show that $\left|\frac{\partial h}{\partial u}\right|$ is also bounded. From (\ref{Dh higher}), we calculate the estimate (see Appendix \ref{Appendix4})
	\begin{eqnarray}\label{estimate dhdu}
	\left|\frac{\partial h}{\partial u}\right|\leq\frac{C(K_1(x)+K_2(x)+K_3+x^3+x^p+x^{p+2})}{(1+u)^{\frac{3(2k-D)}{2}}(1+r+u)^{\frac{3(D-2)}{2}-1}},
	\end{eqnarray}
	where
	\begin{eqnarray}
	K_1(x)&=&C(d+x^3+x^{k+2}+x^p+x^{p+2}+x^{p+4})(1+x^2+x^{k+1}+x^{p+1}+x^{p+3})\label{K1}\nonumber\\
	&&\times\exp\left[C(x^2+x^{k+1}+x^{p+1})\right],\\
	K_2(x)&=&C(x^2+x^{p+1})(d+x^3+x^p+x^{p+2})\exp\left[C(x^2+x^{p+1})\right],\\
	K_3(x)&=&Cx^{p+1}(d+x^3+x^p+x^{p+2})\exp\left[C(x^2+x^{p+1})\right]\label{K3}.
	\end{eqnarray}
	
	We can finally state our main theorem of global existence as follows
	\begin{theorem} \label{Theorem 1}
	Setting $D\geq 4$. Let $X$ and $Y$ be the function spaces defined by (\ref{space X}) and (\ref{space Y}) respectively. For a given an initial data $h(0,r)\in C^1[0,\infty)$ and the positive constant $(D-2)$-spatial Ricci scalar of compact manifold $\hat{R}(\hat{\sigma})$ such that $h(0,r)=O(r^{-{(k-1)}})$ and $\frac{\partial h}{\partial r}(0,r)=O(r^{-k})$ as $r\rightarrow \infty$, there exists a global classical solution of equation (\ref{Dh higher}) for $k>\frac{D}{2}$ and $p\in[k,\infty)$, as follows
		\begin{eqnarray}
		h(u,r)\in C^1 ([0,\infty)\times [0,\infty)). 
		\end{eqnarray}
	\end{theorem}

\section{Completeness Properties of the Spacetime}
\label{sec4}
Let us define a mass-like function for $D\geq 4$ in Bondi coordinates
\begin{align}\label{M}
M:=\frac{r^{D-3}}{2}\left[\frac{2\hat{R}(\hat{\sigma})}{(D-2)^2}-\frac{\tilde{g}}{g}\right].
\end{align}
If we take $D=4$, equation (\ref{M}) reduce to equation (4.3) in \cite{Chris1}. Then, from
(\ref{g}) we know that $g$ is monotonically nondecreasing function of $r$ at each $u$, and also $\tilde{g}\leq g$. Therefore $M$ is nonnegative. Furthermore, $M$ vanishes at $r=0$.
\begin{proposition}\label{Proposition 1}
	Let $D\geq 4$. For a given positive constant $(D-2)$-spatial Ricci scalar of compact manifold $\hat{R}(\hat{\sigma})$, $M$ is a monotonically nondecreasing function of $r$ at each $u$.
\end{proposition}
\begin{proof}
	From (\ref{g}) and (\ref{gtilde}), we get
	\begin{align}
	\frac{\partial g}{\partial r}=&\frac{8\pi g}{(D-2)}\frac{1}{r}\left(h-\frac{(D-2)}{2}\tilde{h}\right)^2,\\
	\frac{\partial \tilde{g}}{\partial r}=&\left[\frac{1}{r}\frac{\hat{R}(\hat{\sigma})}{(D-2)}-\frac{(D-3)}{r}\frac{\tilde{g}}{g}+\frac{16\pi r}{(D-2)}V(\tilde{h})\right]g.
	\end{align}
	Then, we calculate
	\begin{align}\label{dMdr}
	\frac{\partial M}{\partial r}=\frac{(D-4)\hat{R}(\hat{\sigma})}{2(D-2)^2}r^{D-4}+\frac{4\pi}{(D-2)}\left[\frac{\tilde{g}}{g}\left(h-\frac{(D-2)}{2}\tilde{h}\right)^2-2r^2V(\tilde{h})\right]r^{D-4}.
	\end{align}
	We choose that $V(\tilde{h})$ is a negative function, for example $V(\tilde{h})=-\frac{1}{p+1}|\tilde{h}|^{p+1}$ as defined in \cite{Chae}, such that we have $\frac{\partial M}{\partial r}\geq 0$. Thus, $M$ is a monotonically nondecreasing function of $r$ at each $u$.
\end{proof}
\begin{proposition}\label{Proposition 2}
	Let $D\geq 4$. For a given positive constant $(D-2)$-spatial Ricci scalar of compact manifold $\hat{R}(\hat{\sigma})$, $M$ is a monotonically nonincreasing function of $u$.
\end{proposition}
\begin{proof}
	Let us consider the characteristics equation
	$\frac{\mathrm{d}r}{\mathrm{d}u}=-\frac{1}{2}\tilde{g}(u,r).$ Using the chain rule we calculate
	\begin{align}
	\frac{\partial M}{\partial u}=&-\frac{1}{2}\tilde{g}\left[\frac{\partial M}{\partial r}=\frac{(D-4)\hat{R}(\hat{\sigma})}{2(D-2)^2}r^{D-4}+\frac{4\pi}{(D-2)}\left[\frac{\tilde{g}}{g}\left(h-\frac{(D-2)}{2}\tilde{h}\right)^2-2r^2V(\tilde{h})\right]r^{D-4}\right].
	\end{align}
	Using similar proof as Proposition \ref{Proposition 1}, we can easily get $\frac{\partial M}{\partial u}\leq 0$ such that $M$ is a monotonically nonincreasing function of $u$.
\end{proof}

We assume that there exist the initial total mass $M_0$ that is finite such that
\begin{align}
\lim_{r\rightarrow\infty}M(0,r):=M_0.
\end{align}
Since $M(u,r)$ is monotonically nondecreasing with respect to $r$ at each $u$ and also bounded by $M_0$, there exists the total mass at retarted time $u$ defined by $M_1:=M_1(u)$ such that
\begin{align}
\lim_{r\rightarrow\infty}M(u,r):=M_1
\end{align}
for all $u\geq 0$. According to \cite{Bondi}, $M_1$ is defined as the Bondi mass. Furthermore, since the Bondi mass $M_1(u)$ is nonegative monotonically nonincreasing function of $u$, there exist the final Bondi mass $M_2$ such that
\begin{align}\label{final Bondi}
\lim_{u\rightarrow\infty} M_1(u):=M_2.
\end{align}

Following \cite{Chris1}, we shall show
\begin{theorem}
Let $D\geq 4$. For a given positive constant $(D-2)$-spatial Ricci scalar of compact manifold $\hat{R}(\hat{\sigma})$ and the final Bondi mass that is defined in (\ref{final Bondi}), the timelike line $r=r_0$ is complete toward future if $r_0>\frac{M_2(D-2)^2}{\hat{R}(\hat{\sigma})}$.
\end{theorem}
\begin{proof}
From (\ref{gtilde}), we have $\log \tilde{g}(u,r)\rightarrow 0$ as $r\rightarrow\infty$ at each $u$. We obtain
\begin{align}
-\log\tilde{g}(u,r_0)=&\int_{r_0}^{\infty}\frac{1}{\tilde{g}}\frac{\partial \tilde{g}}{\partial r}\mathrm{d}r\nonumber\\
=&\int_{r_0}^{\infty}\left[\frac{2(D-3)M}{r^{D-2}}-\frac{\hat{R}(\hat{\sigma})}{(D-2)}\left(\frac{2(D-3)}{(D-2)}-1\right)\frac{1}{r}+\frac{16\pi r V(\tilde{h})}{(D-2)}\right]\nonumber\\
&\times\left[\frac{2\hat{R}(\hat{\sigma})}{(D-2)^2}-\frac{2M}{r^{D-3}}\right]^{-1}\mathrm{d}r.
\end{align}
Since $V(\tilde{h})$ is a negative function and $M(u,r)\leq M_1(u)$, if $r_0^{D-3}>\frac{M_1(u)(D-2)^2}{\hat{R}(\hat{\sigma})}$ then
\begin{align}
-\log \tilde{g}(u,r_0)\leq& \int_{r_0}^{\infty}\frac{2(D-3)M_1(u)}{r^{D-2}}\left(\frac{2\hat{R}(\hat{\sigma})}{(D-2)^2}-\frac{2M_1(u)}{r^{D-3}}\right)^{-1}\mathrm{d}r\nonumber\\
\leq& -\log \left(\frac{\hat{R}(\hat{\sigma})}{(D-2)^2}-\frac{M_1(u)}{{r_0}^{D-3}}\right).
\end{align}
From (\ref{g}) and (\ref{gtilde}) we have $e^{2F}=g\tilde{g}$. Since $\tilde{g}\leq g$, thus
\begin{align}
e^{F(u,r_0)}\geq \tilde{g}(u,r_0)\geq\frac{\hat{R}(\hat{\sigma})}{(D-2)^2}-\frac{M_1(u)}{{r_0}^{D-3}}.
\end{align}
Let us introduce $e^{F(u,r_0)}\mathrm{d}u$ as the proper time element along the line $r=r_0$. Since $M_1(u)\rightarrow M_2$ for $u\rightarrow \infty$, if $r_0^{D-3}>\frac{M_2(D-2)^2}{\hat{R}(\hat{\sigma})}$ there exists $u_1,u_2$ where $u_2>u_1$ such that
\begin{align}
\int_{u_1}^{u_2}e^{F(u,r_0)}\mathrm{d}u\geq&\int_{u_1}^{u_2}\left(\frac{\hat{R}(\hat{\sigma})}{(D-2)^2}-\frac{M_1(u)}{{r_0}^{D-3}}\right)\mathrm{d}u\nonumber\\
\geq&\left(\frac{\hat{R}(\hat{\sigma})}{(D-2)^2}-\frac{M_2}{{r_0}^{D-3}}\right)(u_2-u_1)\rightarrow \infty
\end{align}
as $u_2\rightarrow\infty$. The above equation ensures that the points of the hypersurfaces are at future timelike infinity for $r_0>\frac{M_2(D-2)^2}{\hat{R}(\hat{\sigma})}$. The completeness of spacetime along the future directed timelike lines outward to a region which resembles the event horizon of the black hole. This is the end of the proof.
\end{proof}
	\section{Appendix}
	\subsection{Estimate for (\ref{fk})}\label{Appendix1}
	We write the estimate for (\ref{fk}) as follows
	\begin{align}
	|f|\leq&\frac{(D-2)}{4r}\left|\frac{\hat{R}(\hat{\sigma})}{(D-2)}g-\frac{(D-2)}{2}\tilde{g}\right|\left|\tilde{h}\right|+\frac{8\pi g r}{(D-2)}\left|V(\tilde{h})\right|\left|\tilde{h}\right|+\frac{gr}{2}\left|\frac{\partial V(\tilde{h})}{\partial \tilde{h}}\right|,\nonumber\\
	\leq& A_1 + A_2 + A_3
	\end{align}
	\begin{enumerate}
		\item Estimate for $A_1$\\
		First, we calculate
		\begin{align}\label{htilde higher}
		|\tilde{h}|\leq \frac{1}{r^{\frac{(D-2)}{2}}}\int_{0}^{r}\frac{\|h\|_X}{(1+s+u)^{k-1}}\frac{1}{s^{\frac{(4-D)}{2}}}~\mathrm{d}s\leq\frac{2x}{(2k-D)}\frac{1}{(1+u)^{\frac{(2k-D)}{2}}(1+r+u)^{\frac{(D-2)}{2}}}.
		\end{align}
		Then, we obtain
		\begin{align}\label{h-h}
		\left|h-\frac{(D-2)}{2}\tilde{h}\right|\leq& |h| + \left|\frac{(D-2)}{2}\tilde{h}\right|\nonumber\\
		\leq&\frac{x}{(1+r+u)^{k-1}}+\frac{2x}{(2k-D)}\frac{1}{(1+u)^{\frac{(2k-D)}{2}}(1+r+u)^{\frac{(D-2)}{2}}}\nonumber\\
		\leq& \frac{Cx}{(1+u)^{\frac{(2k-D)}{2}}(1+r+u)^{\frac{(D-2)}{2}}}.
		\end{align}
		From the above estimate, we get
		\begin{align}
		|g(u,r)-g(u,r')|\leq&\int_{r'}^{r}\left|\frac{\partial g}{\partial s}(u,s) \right|\mathrm{d}s\leq \frac{8\pi}{(D-2)} \int_{r'}^{r}\frac{1}{s}\left|h-\frac{(D-2)}{2}\tilde{h}\right|^2\mathrm{d}s\nonumber\\
		\leq& \frac{8\pi}{(D-2)^2}\frac{x^2}{(1+u)^{2k-D}}\left[\frac{1}{(1+r'+u)^{D-2}}-\frac{1}{(1+r+u)^{D-2}}\right],
		\end{align}
		and
		\begin{align}\label{g-gbar k}
		|(g-\bar{g})(u,r)|\leq& \frac{1}{r}\int_{0}^{r}|g(u,r)-g(u,r')|\mathrm{d}r'
		\nonumber\\
		\leq& \frac{8\pi x^2}{(D-3)(D-2)^2}\frac{1}{(1+u)^{2k-3}(1+r+u)}.
		\end{align}
		Thus
		\begin{align}
		\int_{r}^{\infty}\frac{1}{s}\left|h-\frac{(D-2)}{2}\tilde{h}\right|^2\mathrm{d}s\leq\frac{x^2}{(D-2)(1+u)^{2k-D}(1+r+u)^{D-2}},
		\end{align}
		such that we obtain
		\begin{align}
		|g (u,r)|=& \exp \left[-\frac{8\pi}{(D-2)}\int_{r}^{\infty}\frac{1}{s}\left|h-\frac{(D-2)}{2}\tilde{h}\right|^2\mathrm{d}s\right]\nonumber\\
		\geq&\exp \left[-\frac{8\pi x^2}{(D-2)^2(1+u)^{2k-D}(1+r+u)^{D-2}}\right].
		\end{align}
		As a consequence, we can calculate
		\begin{align}
		|\bar{g}|\geq |g| + |g-\bar{g}|\geq\frac{8\pi x^2}{(D-3)(D-2)^2}\frac{1}{(1+u)^{2k-3}(1+r+u)},
		\end{align}
		and
		\begin{align}
		\frac{(D-4)}{2}\frac{\hat{R}(\hat{\sigma})}{r^{D-3}}\int_{0}^{r}|\bar{g}|s^{D-4}\mathrm{d}s\leq \frac{Cx^2}{(1+u)^{2k-3}(1+r+u)}.
		\end{align}
		To proceed, we calculate
		\begin{align}
		\frac{8\pi}{r^{D-3}}\int_{0}^{r} gs^{D-2}|V(\tilde{h})|\mathrm{d}s\leq \frac{Cx^{p+1}}{(1+u)^{k^2-D}(1+r+u)^{D-3}}.
		\end{align}
		Now, we obtain
		\begin{align}
		\left|\frac{\hat{R}(\hat{\sigma})}{(D-2)}g-\frac{(D-2)}{2}\tilde{g}\right|\leq&|g-\bar{g}|	+\frac{(D-4)}{2}\frac{\hat{R}(\hat{\sigma})}{r^{D-3}}\int_{0}^{r}|\bar{g}|s^{D-4}\mathrm{d}s\nonumber\\
		&+\frac{8\pi}{r^{D-3}}\int_{0}^{r} gs^{D-2}|V(\tilde{h})|\mathrm{d}s\nonumber\\
		\leq& \frac{C(x^2+x^{p+1})}{(1+u)^{2k-3}(1+r+u)}.
		\end{align}
		Hence,
		\begin{align}
		A_1=\frac{(D-2)}{4r}\left|\frac{\hat{R}(\hat{\sigma})}{(D-2)}g-\frac{(D-2)}{2}\tilde{g}\right|\left|\tilde{h}\right|\leq\frac{C(x^3+x^{p+2})}{(1+u)^{\frac{6(k-1)-D}{2}}(1+r+u)^{\frac{D+2}{2}}}.
		\end{align}
		\item Estimate for $A_2$\\
		Using (\ref{htilde}) and (\ref{estimasi Potensial}), we obtain
		\begin{align}
		A_2=\frac{8\pi g r}{(D-2)}\left|V(\tilde{h})\right||\tilde{h}|\leq\frac{Cx^{p+2}}{(1+u)^{\frac{(2k-D)}{2}(k+2)}(1+r+u)^{\frac{(D-2)(k+2)-2}{2}}}.
		\end{align}
		\item Estimate for $A_3$\\
		Again, using (\ref{estimasi Potensial}) we get
		\begin{align}
		A_3=\frac{gr}{2}\left|\frac{\partial V(\tilde{h})}{\partial \tilde{h}}\right|\leq\frac{Cx^p}{(1+u)^{\frac{(2k-D)k}{2}}(1+r+u)^{\frac{(D-2)k-2}{2}}}.
		\end{align}
	\end{enumerate}
	Finally, we obtain
	\begin{align}
	|f|\leq \frac{C(x^3+x^p+x^{p+2})}{(1+u)^{\frac{(2k-D)k}{2}}(1+r+u)^{\frac{(D-2)k-2}{2}}}.
	\end{align}
	
	\subsection{Estimate for (\ref{f1 higher})}\label{Appendix2}
	We write the estimate of (\ref{f1 higher}) as follows:
	\begin{align}
	|f_1|\leq& \left\{\left|\frac{1}{2r}\left(\frac{\hat{R}(\hat{\sigma})}{(D-2)}\frac{\partial g}{\partial r}-\frac{(D-2)}{2}\frac{\partial\tilde{g}}{\partial r}\right)\right|+\left|\frac{1}{2r^2}\left(\frac{\hat{R}(\hat{\sigma})}{(D-2)}g-\frac{(D-2)}{2}\tilde{g}\right)\right|\right.\nonumber\\
	&\left.+\left|\frac{8\pi g r}{(D-2)} \frac{\partial V(\tilde{h})}{\partial \tilde{h}} \frac{\partial \tilde{h}}{\partial r}\right|+\left|\frac{8\pi gV(\tilde{h})}{(D-2)}\right|+\left|\frac{8\pi rV(\tilde{h})}{(D-2)}\frac{\partial g}{\partial r}\right|\right\}\left(\left|\mathcal{F}\right|+\frac{(D-2)}{2}\left|\tilde{h}\right|\right)\nonumber\\
	& +\left\{\left|\frac{1}{2r} \left(\frac{\hat{R}(\hat{\sigma})}{(D-2)}g-\frac{(D-2)}{2}\tilde{g}\right)\right|+\left|\frac{8\pi g rV(\tilde{h})}{(D-2)}\right|+\left|\frac{g r }{2}\frac{\partial^2V(\tilde{h})}{\partial\tilde{h}^2}\right|\right\}\left|\frac{\partial \tilde{h}}{\partial r}\right|+\left|\frac{r}{2}\frac{\partial g}{\partial r}+\frac{g}{2}\right|\left|\frac{\partial V(\tilde{h})}{\partial\tilde{h}}\right|\nonumber\\
	=& (B_1+B_2+B_3+B_4+B_5)\left(\left|\mathcal{F}\right|+\frac{(D-2)}{2}\left|\tilde{h}\right|\right) + (B_6+B_7+B_8)\left|\frac{\partial \tilde{h}}{\partial r}\right| + B_9
	\end{align}
	\begin{enumerate}
		\item Estimate for $B_1$\\
		We use the formula $\frac{\partial\bar{g}}{\partial r}=\frac{g-\bar{g}}{r}$ to obtain
		\begin{align}
		\left|\frac{\partial \tilde{g}}{\partial r}\right|\leq&\frac{\hat{R}(\hat{\sigma})}{(D-2)}\frac{|g-\bar{g}|}{r}+\frac{\hat{R}(\hat{\sigma})(D-4)(D-3)}{(D-2)r^{D-2}}\int_{0}^{r}|\bar{g}|s^{D-4}\mathrm{d}s+\frac{\hat{R}(\hat{\sigma})(D-4)}{(D-2)}\frac{|\bar{g}|}{r}\nonumber\\
		&+\frac{(D-3)16\pi}{(D-2)r^{D-2}}\int_{0}^{r} \left|g\right|s^{D-2}\left|V(\tilde{h})\right|\mathrm{d}s+\frac{16\pi |g|r\left|V(\tilde{h})\right|}{(D-2)}.
		\end{align}
		The estimate (\ref{g-gbar k}) yields
		\begin{align}
		\frac{|g-\bar{g}|}{r}\leq\frac{Cx^2r^{D-4}}{(1+u)^{2k-3}(1+r+u)^{D-2}}.
		\end{align}
		Using (\ref{estimate gbar}), we get
		\begin{align}
		\frac{\hat{R}(\hat{\sigma})(D-4)(D-3)}{(D-2)r^{D-2}}\int_{0}^{r}|\bar{g}|s^{D-4}\mathrm{d}s\leq\frac{Cx^2}{(1+u)^{2k-3}(1+r+u)^2},
		\end{align}
		and
		\begin{align}
		\frac{(D-4)}{(D-2)}\hat{R}(\hat{\sigma})\frac{1}{r}|\bar{g}|\leq\frac{8\pi x^2}{(D-3)(D-2)^2(1+u)^{2k-3}(1+r+u)^2}.
		\end{align}
		Then, using (\ref{htilde}) we have
		\begin{align}
		\frac{(D-3)16\pi}{(D-2)r^{D-2}}\int_{0}^{r} \left|g\right|s^{D-2}\left|V(\tilde{h})\right|\mathrm{d}s\leq \frac{Cx^{p+1}}{(1+u)^{k^2-D}(1+r+u)^{D-2}},
		\end{align}
		and
		\begin{align}
		\frac{16\pi gr}{(D-2)}|V(\tilde{h})|\leq\frac{Cx^{p+1}r}{(1+u)^{\frac{(2k-D)(k+1)}{2}}(1+r+u)^{\frac{(D-2)(k+1)}{2}}}.
		\end{align}
		Thus,
		\begin{align}
		\frac{(D-2)}{2}\left|\frac{\partial\tilde{g}}{\partial r}\right|\leq \frac{C(x^2+x^{k+1}+x^{p+1})r}{(1+u)^{2k-3}(1+r+u)^3}.
		\end{align}
		Using (\ref{g}) we calculate
		\begin{align}\
		\frac{\hat{R}(\hat{\sigma})}{(D-2)}\left|\frac{\partial g}{\partial r}\right|\leq\frac{Cx^2r}{(1+u)^{2k-D}(1+r+u)^{D}}.
		\end{align}
		Finally, we obtain
		\begin{align}
		B_1=\left|\frac{1}{2r}\left(\frac{\hat{R}(\hat{\sigma})}{(D-2)}\frac{\partial g}{\partial r}-\frac{(D-2)}{2}\frac{\partial\tilde{g}}{\partial r}\right)\right|\leq\frac{C(x^2+x^{k+1}+x^{p+1})}{(1+u)^{2k-D}(1+r+u)^3}.
		\end{align}
		
		\item Estimate for $B_2$\\
		Using (\ref{g-gtilde k}) we obtain
		\begin{align}
		B_2=\left|\frac{1}{2r^2}\left(\frac{\hat{R}(\hat{\sigma})}{(D-2)}g-\frac{(D-2)}{2}\tilde{g}\right)\right|\leq\frac{C(x^2+x^{p+1})}{(1+u)^{2k-3}(1+r+u)^{3}}.
		\end{align}
		
		\item Estimate for $B_3$\\
		Using the relation $\left|\frac{\partial \tilde{h}}{\partial r}\right|=\frac{\left|h-\frac{(D-2)}{2}\tilde{h}\right|}{r}$, we get
		\begin{align}
		B_3 = \left|\frac{8\pi g r}{(D-2)} \frac{\partial V(\tilde{h})}{\partial \tilde{h}} \frac{\partial \tilde{h}}{\partial r}\right|\leq\frac{Cx^{p+1}}{(1+u)^{\frac{(2k-D)(k+1)}{2}}(1+r+u)^{\frac{(D-2)(k+1)}{2}}}.
		\end{align}
		
		\item Estimate for $B_4$\\
		Using (\ref{estimasi Potensial}) we obtain
		\begin{align}
		B_4=\left|\frac{8\pi g V(\tilde{h})}{(D-2)}\right|
		\leq \frac{C x^{p+1}}{(1+u)^{\frac{(2k-D)(k+1)}{2}}(1+r+u)^{\frac{(D-2)(k+1)}{2}}}.
		\end{align}
		
		\item Estimate for $B_5$\\
		Combination of (\ref{estimasi Potensial}) and (\ref{gr higher}) yields
		\begin{align}
		B_5 = \left|\frac{8\pi rV(\tilde{h})}{(D-2)}\frac{\partial g}{\partial r}\right|\leq\frac{Cx^{p+3}}{(1+u)^{3(2k-D)}(1+r+u)^{\frac{(D-2)(k+3)}{2}}}.
		\end{align}
		
		\item Estimate for $B_6$\\
		Using (\ref{g-gtilde k}) we obtain
		\begin{align}
		B_6=\left|\frac{1}{2r} \left(\frac{\hat{R}(\hat{\sigma})}{(D-2)}g-\frac{(D-2)}{2}\tilde{g}\right)\right|\leq\frac{C(x^2+x^{p+1})r}{(1+u)^{2k-3}(1+r+u)^{3}}.
		\end{align}
		
		\item Estimate for $B_7$\\
		The estimate (\ref{estimasi Potensial}) produces
		\begin{align}
		B_7=\left|\frac{8\pi g r V(\tilde{h})}{(D-2)}\right|\leq\frac{C x^{p+1}r}{(1+u)^{\frac{(2k-D)(k+1)}{2}}(1+r+u)^{\frac{(D-2)(k+1)}{2}}}.
		\end{align}
		
		\item Estimate for $B_8$\\
		Using (\ref{estimasi Potensial}) we have
		\begin{align}
		B_8 =\left|\frac{g r }{2}\frac{\partial^2V(\tilde{h})}{\partial\tilde{h}^2}\right|\leq \frac{C x^{p-1}r}{(1+u)^{\frac{(2k-D)(k-1)}{2}}(1+r+u)^{\frac{(D-2)(k-1)}{2}}}.
		\end{align}
		\item Estimate for $B_9$\\
		Again, using (\ref{estimasi Potensial}) we get
		\begin{align}
		B_9 =\left|\frac{r}{2}\frac{\partial g}{\partial r}+\frac{g}{2}\right|\left|\frac{\partial V(\tilde{h})}{\partial\tilde{h}}\right|\leq\frac{C x^{p+2}r}{(1+u)^{\frac{(2k-D)(k+1)}{2}}(1+r+u)^{\frac{(D-2)(k+2)+2}{2}}}.
		\end{align}
	\end{enumerate}
	Finally, we obtain
	\begin{align}
	|f_1|\leq\frac{C(x^2+x^{k+1}+x^{p+1}+x^{p+3})}{(1+u)^{2k-D}(1+r+u)^3}|\mathcal{F}|+\frac{C(x^3+x^{k+2}+x^{p+2}+x^{p+4})}{(1+u)^{2k-3}(1+r+u)^{\frac{D+4}{2}}}+\frac{C(x^3+x^{p}+x^{p+2})r}{(1+u)^{2k-3}(1+r+u)^{\frac{D+4}{2}}}.
	\end{align}
	
	\subsection{Estimate for (\ref{estimate psi2})}\label{Appendix3}
	We write the estimate of (\ref{estimate psi2}) as follows:
	\begin{align}
	|\tilde{\psi}|\leq&\frac{1}{2}\left|\tilde{g}_1-\tilde{g}_2\right||\mathcal{G}_2|+\frac{1}{2r}\left|\frac{\hat{R}(\hat{\sigma})}{(D-2)}g_1-\frac{(D-2)}{2}\tilde{g}_1\right|\frac{(D-2)}{2}\left|\tilde{h}_1 -\tilde{h}_2\right|\nonumber\\
	&+\frac{1}{2r}\left|\frac{\hat{R}(\hat{\sigma})}{(D-2)}(g_1-g_2)-\frac{(D-2)}{2}(\tilde{g}_1-\tilde{g}_2)\right||\mathcal{F}_2|+\frac{8\pi r g_2}{(D-2)}\left|V(\tilde{h}_1)-V(\tilde{h}_1)\right||\mathcal{F}_2|\nonumber\\
	&+\frac{1}{2r}\left|\frac{\hat{R}(\hat{\sigma})}{(D-2)}(g_1-g_2)-\frac{(D-2)}{2}(\tilde{g}_1-\tilde{g}_2)\right|\frac{(D-2)}{2}|\tilde{h}_2|+\frac{8\pi r|V(\tilde{h}_1)|}{(D-2)}|g_1-g_2||\mathcal{F}_1|\nonumber\\
	&+4\pi r|g_1-g_2||\tilde{h}_1| |V(\tilde{h}_1)|+4\pi r g_2|\tilde{h}_1 - \tilde{h}_2||V(\tilde{h}_2)|+4\pi r g_2 |\tilde{h}_1|\left|V(\tilde{h}_1)-V(\tilde{h}_2)\right|\nonumber\\
	&+\frac{r}{2}|g_1-g_2||\tilde{h}_1|\left|\frac{\partial^2V(\tilde{h}_1)}{\partial\tilde{h}_1^2 }\right|+\frac{r}{2}g_2|\tilde{h}_1-\tilde{h}_2|\left|\frac{\partial^2V(\tilde{h}_1)}{\partial\tilde{h}_1^2 }\right|+\frac{r}{2}g_2|\tilde{h}_1|\left|\frac{\partial^2V(\tilde{h}_1)}{\partial\tilde{h}_1^2 }-\frac{\partial^2V(\tilde{h}_2)}{\partial\tilde{h}_2^2 }\right|\nonumber\\
	:=& D_1+D_2+D_3+...+D_{10}+D_{11}+D_{12}
	\end{align}
	\begin{enumerate}
		\item Estimate for $D_1$\\
		Setting $\|h_1-h_2\|_Y=y$. Using (\ref{htilde}) we obtain
		\begin{align}\label{h1-h2}
		|\tilde{h}_1 - \tilde{h}_2|\leq\frac{2y}{(2k-D)}\frac{1}{(1+u)^{\frac{(2k-D)}{2}}(1+r+u)^{\frac{(D-2)}{2}}}.
		\end{align}
		To proceed, we calculate
		\begin{align}
		|h_1-h_2-(\tilde{h}_1-\tilde{h}_2)|\leq& |h_1-h_2|+|\tilde{h}_1-\tilde{h}_2|\nonumber\\
		\leq&\frac{Cy}{(1+u)^{\frac{(2k-D)}{2}}(1+r+u)^{\frac{(D-2)}{2}}},
		\end{align}
		and
		\begin{align}
		\left||h_1-\tilde{h}_1|^2-|h_2-\tilde{h}_2|^2\right|\leq& \left|(h_1-h_2)-(\tilde{h}_1-\tilde{h}_2) \right|\left(|h_1-\tilde{h}_1|+|h_2-\tilde{h}_2|\right)\nonumber\\
		\leq& \frac{Cxy}{(1+u)^{2k-D}(1+r+u)^{D-2}}.
		\end{align}
		As a consequence, we get
		\begin{align}
		|g_1-g_2|\leq&\frac{8\pi}{(D-2)} \int_{r}^{\infty}\frac{1}{s}\left||h_1-\tilde{h}_1|^2-|h_2-\tilde{h}_2|^2\right|\mathrm{d}s\nonumber\\
		\leq& \frac{Cxy}{(1+u)^{(2k-D)}(1+r+u)^{(D-2)}}.
		\end{align}
		Then, we calculate
		\begin{align}
		\tilde{h}_1^{p+1} - \tilde{h}_2^{p+1}=(\tilde{h}_1-\tilde{h}_2)\int_{0}^{1}\left(t\tilde{h}_1+(1-t)\tilde{h}_2\right)\left|t\tilde{h}_1+(1-t)\tilde{h}_2\right|^{p-1}\mathrm{d}t,
		\end{align}
		yields
		\begin{align}
		\left|\tilde{h}_1^{p+1}-\tilde{h}_2^{p+1}\right|\leq& |\tilde{h}_1-\tilde{h}_2|\left(|\tilde{h}_1|+|\tilde{h}_2|\right)^p\nonumber\\
		\leq&\frac{2^{2p+1}x^py}{(2k-D)^{k+1}(1+u)^{\frac{(2k-D)(k+1)}{2}}(1+r+u)^{\frac{(D-2)(k+1)}{2}}}.
		\end{align}
		We estimate
		\begin{align}
		\frac{16\pi}{(D-2)}\frac{1}{r^{D-3}}\int_{0}^{r} g_1s^{D-2}\left|V(\tilde{h}_1)-V(\tilde{h}_2)\right|\mathrm{d}s\leq\frac{Cx^py}{(1+u)^{k^2-D}(1+r+u)^{D-3}}.
		\end{align}
		Then, we obtain
		\begin{align}
		|\bar{g}_1-\bar{g}_2|\leq\frac{1}{r}\int_{0}^{r}|g_1-g_2|\mathrm{d}s\leq\frac{Cxy}{(1+u)^{2k-3}(1+r+u)},
		\end{align}
		and
		\begin{align}
		\frac{(D-4)}{(D-2)}\frac{\hat{R}(\hat{\sigma})}{r^{D-3}}\int_{0}^{r}\left|\bar{g}_1-\bar{g}_2\right|s^{D-4}\mathrm{d}s\leq \frac{Cxy}{(1+u)^{2k-3}(1+r+u)}.
		\end{align}
		From the definition (\ref{gtilde}) we get
		\begin{align}
		\left|\tilde{g}_1-\tilde{g}_2\right|\leq& \frac{\hat{R}(\hat{\sigma})}{(D-2)}\left|\bar{g}_1-\bar{g}_2\right|+\frac{(D-4)}{(D-2)}\frac{\hat{R}(\hat{\sigma})}{r^{D-3}}\int_{0}^{r}\left|\bar{g}_1-\bar{g}_2\right|s^{D-4}\mathrm{d}s\nonumber\\
		&+ \frac{16\pi}{(D-2)}\frac{1}{r^{D-3}}\int_{0}^{r} g_1s^{D-2}\left|V(\tilde{h}_1)-V(\tilde{h}_2)\right|\mathrm{d}s\nonumber\\
		\leq&\frac{Cy(x+x^p)}{(1+u)^{2k-3}(1+r+u)}.
		\end{align}
		Finally, we obtain
		\begin{align}
		D_1=&\frac{1}{2}|\tilde{g}_1-\tilde{g}_2||\mathcal{G}_2|\leq\frac{Cy(x+x^p)}{(1+u)^{2k-3}(1+r+u)}\frac{C\exp\left[C(x^2+x^{k+1}+x^{p+1})\right]}{\kappa^k(1+r_1+u_1)^k}\nonumber\\
		&\times(d+x^3+x^{k+2}+x^p+x^{p+2}+x^{p+4})(1+x^2+x^{k+1}+x^{p+1}+x^{p+3}).
		\end{align}
		Since $1+r(u)+u\geq \frac{1}{2}\kappa(1+r_1+u_1)$, we have
		\begin{align}
		D_1\leq \frac{Cy\alpha(x)}{(1+u)^{2k-3}(1+r+u)^{k+1}}
		\end{align}
		where $\alpha(x)=C(x+x^p)\left(d+x^3+x^{k+2}+x^p+x^{p+2}+x^{p+4}\right)(1+x^2+x^{k+1}+x^{p+1}+x^{p+3})\exp\left[C(x^2+x^{k+1}+x^{p+1})\right].$
		\item Estimate for $D_2$\\
		Combining (\ref{g-gtilde k}) and (\ref{h1-h2}) we obtain
		\begin{align}
		D_2=\frac{1}{2r}\left|\frac{\hat{R}(\hat{\sigma})}{2}g_1-\frac{(D-2)^2}{4}\tilde{g}_1\right||\tilde{h}_1-\tilde{h}_2|\leq \frac{C(x^2+x^{p+1})y}{(1+u)^{\frac{6(k-1)-D}{2}}(1+r+u)^{\frac{D+2}{2}}}.
		\end{align}
		\item Estimate for $D_3$\\
		We firstly calculate
		\begin{align}
		\left|\frac{\hat{R}(\hat{\sigma})}{(D-2)}(g_1-g_2)-\frac{\hat{R}(\hat{\sigma})}{2}(\bar{g}_1-\bar{g}_2)\right|\leq&\left|g_1-g_2-(\bar{g}_1-\bar{g}_2) \right|
		\leq\frac{1}{r}\int_{0}^{r}\int_{r'}^{r}\left|\frac{\partial}{\partial r}(g_1-g_2)\right|\mathrm{d}s~\mathrm{d}r'\nonumber\\
		\leq&\frac{8\pi}{(D-2)r}\int_{0}^{r}\int_{r'}^{r}\frac{1}{s}\left||h_1-\tilde{h}_1|^2-|h_2-\tilde{h}_2|^2\right|\mathrm{d}s~\mathrm{d}r'\nonumber\\
		\leq& \frac{Cxy}{r(1+u)^{2k-D}}\int_{0}^{r}\int_{r'}^{r}\frac{\mathrm{d}s}{(1+s+u)^{D-1}}\mathrm{d}r'\nonumber\\
		\leq& \frac{Cxyr}{(1+u)^{2k-3}(1+r+u)^{D-2}}.
		\end{align}
		Thus,
		\begin{align}\label{g-gtilde}
		&\frac{1}{2r}\left|\frac{\hat{R}(\hat{\sigma})}{(D-2)}(g_1-g_2)-\frac{(D-2)}{2}(\tilde{g}_1-\tilde{g}_2)\right|\leq\frac{1}{2r}\left[	\left|\frac{\hat{R}(\hat{\sigma})}{(D-2)}(g_1-g_2)-\frac{\hat{R}(\hat{\sigma})}{2}(\bar{g}_1-\bar{g}_2)\right|\right.\nonumber\\
		&\left.+\frac{(D-4)}{2}\frac{\hat{R}(\hat{\sigma})}{r^{D-3}}\int_{0}^{r}\left|\bar{g}_1-\bar{g}_2\right|s^{D-4}\mathrm{d}s+\frac{8\pi}{r^{D-3}}\int_{0}^{r} g_1s^{D-2}\left|V(\tilde{h}_1)-V(\tilde{h}_2)\right|\mathrm{d}s\right]\nonumber\\
		&\leq\frac{C(x+x^p)y}{(1+u)^{2k-3}(1+r+u)^{D-2}}.
		\end{align}
		Finally, we obtain
		\begin{align}
		D_3=&\frac{1}{2r}\left|\frac{\hat{R}(\hat{\sigma})}{(D-2)}(g_1-g_2)-\frac{(D-2)}{2}(\tilde{g}_1-\tilde{g}_2)\right||\mathcal{F}_2|\nonumber\\
		\leq& \frac{C(x+x^p)y}{(1+u)^{2k-3}(1+r+u)^{D-2}}\frac{C(d+x^3+x^p+x^{p+2})\exp\left[C(x^2+x^{p+1})\right]}{\kappa^{k-1}(1+r_1+u_1)^{k-1}}.
		\end{align}
		Since $1+r(u)+u\geq \frac{1}{2}\kappa(1+r_1+u_1)$, we have
		\begin{align}
		D_3\leq\frac{C y \beta(x)}{(1+u)^{2k-3}(1+r+u)^{D+k-3}},
		\end{align}
		where $\beta(x)=C(x+x^p)(d+x^3+x^p+x^{p+2})\exp\left[C(x^2+x^{p+1})\right].$
		
		\item Estimate for $D_4$\\
		Using (\ref{estimasi Potensial}) and (\ref{estimate Fcurl}) we obtain
		\begin{align}
		D_4=&\frac{8\pi r g_2}{(D-2)}	\left|V(\tilde{h}_1)-V(\tilde{h}_2)\right||\mathcal{F}_2|\nonumber\\
		\leq& \frac{Crx^py}{(1+u)^{\frac{(2k-D)(k+1)}{2}}(1+r+u)^{\frac{(D-2)(k+1)}{2}}}\frac{C(d+x^3+x^p+x^{p+2})\exp\left[C(x^2+x^{p+1})\right]}{\kappa^{k-1}(1+r_1+u_1)^{k-1}}.
		\end{align}
		Since $1+r(u)+u\geq \frac{1}{2}\kappa(1+r_1+u_1)$, we have
		\begin{align}
		D_4\leq\frac{Cy \gamma(x)}{(1+u)^{2(2k-D)}(1+r+u)^{2(D-2)+1}},
		\end{align}
		where $\gamma(x)=C x^p(d+x^3+x^p+x^{p+2})\exp[C(x^2+x^{p+1})].$
		
		\item Estimate for $D_5$\\
		Combination of (\ref{estimate htilde}) and (\ref{g-gtilde}) yields
		\begin{align}
		D_5=\frac{1}{2r}\left|\frac{\hat{R}(\hat{\sigma})}{2}(g_1-g_2)-\frac{(D-2)^2}{4}(\tilde{g}_1-\tilde{g}_2)\right||\tilde{h}_2|
		\leq\frac{C(x^2+x^{p+1})y}{(1+u)^{\frac{6(k-1)-D}{2}}(1+r+u)^{\frac{3(D-2)}{2}}}.
		\end{align}
		\item Estimate for $D_6$\\
		We calculate
		\begin{align}
		D_6=&\frac{8\pi r}{(D-2)}|g_1-g_2||V(\tilde{h})||\mathcal{F}_1|\nonumber\\
		\leq&\frac{8\pi r}{(D-2)}\frac{xy}{(1+u)^{(2k-D)}(1+r+u)^{(D-2)}}\frac{K_0x^{p+1}}{(1+u)^{\frac{(2k-D)(k+1)}{2}}(1+r+u)^{\frac{(D-2)(k+1)}{2}}}\nonumber\\
		&\times\frac{C(d+x^3+x^p+x^{p+2})\exp\left[C(x^2+x^{p+1})\right]}{\kappa^{k-1}(1+r_1+u_1)^{k-1}}.
		\end{align}
		Since $1+r(u)+u\geq \frac{1}{2}\kappa(1+r_1+u_1)$, we obtain
		\begin{align}
		D_6\leq\frac{Cy\sigma(x)}{(1+u)^{3(2k-D)}(1+r+u)^{3(D-2)+k-2}},
		\end{align}
		where $\sigma(x)=Cx^{p+2}(d+x^3+x^p+x^{p+2})\exp[C(x^2+x^{p+1})].$
		
		\item Estimate for $D_7$\\
		We use (\ref{estimasi Potensial}) and (\ref{estimate htilde}) to obtain
		\begin{align}
		D_7 = 4\pi r|g_1-g_2||\tilde{h}_1| |V(\tilde{h}_1)|	\leq\frac{Cyx^{p+3}}{(1+u)^{\frac{7(2k-D)}{2}}(1+r+u)^{\frac{6(2k-D)}{2}}}.
		\end{align}
		\item Estimate for $D_8$\\
		Combining (\ref{estimasi Potensial}) and (\ref{estimate h1-h2}) such that
		\begin{align}
		D_8 = 4\pi r g_2|\tilde{h}_1 - \tilde{h}_2||V(\tilde{h}_2)\leq\frac{Cyx^{p+1}}{(1+u)^{\frac{5(2k-D)}{2}}(1+r+u)^{2(2k-D)}}.
		\end{align}
		\item Estimate for $D_9$\\
		We calculate
		\begin{align}
		D_9 = 4\pi r g_2 |\tilde{h}_1|\left|V(\tilde{h}_1)-V(\tilde{h}_2)\right|\leq\frac{Cyx^{p+1}}{(1+u)^{\frac{5(2k-D)}{2}}(1+r+u)^{\frac{4(2k-D)}{2}}}.
		\end{align}
		\item Estimate for $D_{10}$\\
		Combination of (\ref{estimasi Potensial}) and (\ref{estimate htilde}) yields
		\begin{align}
		D_{10}=\frac{r}{2}|g_1-g_2||\tilde{h}_1|\left|\frac{\partial^2V(\tilde{h}_1)}{\partial\tilde{h}_1^2 }\right|\leq\frac{Cyx^{p+1}}{(1+u)^{\frac{5(2k-D)}{2}}(1+r+u)^{2(2k-D)}}.
		\end{align}
		\item Estimate for $D_{11}$\\
		Again, we use (\ref{estimasi Potensial}) and (\ref{estimate h1-h2}) such that
		\begin{align}
		D_{11} = \frac{r}{2}g_2|\tilde{h}_1-\tilde{h}_2|\left|\frac{\partial^2V(\tilde{h}_1)}{\partial\tilde{h}_1^2 }\right|\leq\frac{Cyx^{p-1}}{(1+u)^{\frac{3(2k-D)}{2}}(1+r+u)^{D-2}}.
		\end{align}
		\item Estimate for $D_{12}$\\
		We calculate
		\begin{align}
		\left|\tilde{h}_1^{p-1}-\tilde{h}_2^{p-1}\right|\leq |\tilde{h}_1-\tilde{h}_2|\left(|\tilde{h}_1|+|\tilde{h}_2|\right)^{p-2}\leq \frac{Cyx^{p-2}}{(1+u)^{\frac{(2k-D)(k-1)}{2}}(1+r+u)^{\frac{(D-2)(k-1)}{2}}}.
		\end{align}
		Hence,
		\begin{align}
		D_{12}=\frac{r}{2}g_2|\tilde{h}_1|\left|\frac{\partial^2V(\tilde{h}_1)}{\partial\tilde{h}_1^2 }-\frac{\partial^2V(\tilde{h}_2)}{\partial\tilde{h}_2^2 }\right| \leq\frac{Cyx^{p+1}}{(1+u)^{\frac{3(2k-D)}{2}}(1+r+u)^{D-2}}.
		\end{align}
	\end{enumerate}
	
	Finally, we obtain
	\begin{align}
	|\tilde{\psi}|\leq\frac{Cy\left[\alpha(x)+\beta(x)+\gamma(x)+\sigma(x)+x^2+x^{p-1}+x^{p+1}+x^{p+3}\right]}{(1+u)^{\frac{3(2k-D)}{2}}(1+r+u)^{D-2}}
	\end{align}
	where $\alpha(x),\beta(x),\gamma(x),$ and $\sigma(x)$ are defined in (\ref{alpha})-(\ref{sigma}) respectively.
	
	\subsection{Estimate for (\ref{estimate dhdu})}\label{Appendix4}
	We write the estimate of (\ref{estimate dhdu}) as follows
	\begin{align}
	\left|\frac{\partial h}{\partial u}\right|\leq& \frac{\tilde{g}}{2}\left|\frac{\partial h}{\partial r}\right| + \frac{1}{2r}\left|\frac{\hat{R}(\hat{\sigma})}{(D-2)}g-\frac{(D-2)}{2}\tilde{g}\right|\left|h\right|+\frac{(D-2)}{4r}\left|\frac{\hat{R}(\hat{\sigma})}{(D-2)}g-\frac{(D-2)}{2}\tilde{g}\right|\left|\tilde{h}\right|\nonumber\\
	&+\frac{8\pi gr}{(D-2)}\left|h\right||V(\tilde{h})|+4\pi gr|\tilde{h}||V(\tilde{h})|+\frac{gr}{2}\left|\frac{\partial V(\tilde{h})}{\partial \tilde{h}}\right|\nonumber\\
	=& H_1+H_2+H_3+H_4+H_5+H_6.
	\end{align}
	\begin{enumerate}
		\item Estimate for $H_1$\\
		Using (\ref{estimate Gcurl}) we obtain
		\begin{align}
		\frac{\tilde{g}}{2}\left|\frac{\partial h}{\partial r}\right|\leq\frac{K_1(x)}{\kappa^{k}(1+r_1+u_1)^{k}}\leq \frac{K_1(x)}{(1+r+u)^{k}}
		\end{align}
		where $K_1(x)=C(d+x^3+x^{k+2}+x^p+x^{p+2}+x^{p+4})(1+x^2+x^{k+1}+x^{p+1}+x^{p+3})\exp\left[C(x^2+x^{k+1}+x^{p+1})\right].$
		
		\item Estimate for $H_2$\\
		Using (\ref{estimate Fcurl}) we obtain
		\begin{align}
		\frac{1}{2r}\left|\frac{\hat{R}(\hat{\sigma})}{(D-2)}g-\frac{(D-2)}{2}\tilde{g}\right|\left|h\right|\leq\frac{K_2(x)}{(1+u)^{2k-3}(1+r+u)^{k+1}},
		\end{align}
		where $K_2(x)=C(x^2+x^{p+1})(d+x^3+x^p+x^{p+2})\exp\left[C(x^2+x^{p+1})\right].$
		
		\item Estimate for $H_3$\\
		We use (\ref{estimate htilde}) and (\ref{g-gtilde k}) to obtain
		\begin{align}
		\frac{(D-2)}{4r}\left|\frac{\hat{R}(\hat{\sigma})}{(D-2)}g-\frac{(D-2)}{2}\tilde{g}\right|\left|\tilde{h}\right|\leq\frac{C(x^3+x^{p+2})}{(1+u)^{\frac{6(k-1)-D}{2}}(1+r+u)^{\frac{D+2}{2}}}.
		\end{align}
		
		\item Estimate for $H_4$\\
		Combination of (\ref{estimasi Potensial}) and (\ref{estimate Fcurl}) yields
		\begin{align}
		\frac{8\pi gr}{(D-2)}\left|h\right||V(\tilde{h})|\leq\frac{K_3(x)}{(1+u)^{2(2k-D)}(1+r+u)^{2(2k-D)+k-2}},
		\end{align}
		where $K_3(x)=Cx^{p+1}(d+x^3+x^p+x^{p+2})\exp\left[C(x^2+x^{p+1})\right].$
		
		\item Estimate for $H_5$\\
		Combination of (\ref{estimasi Potensial}) and (\ref{estimate htilde}) yields
		\begin{align}
		4\pi gr|\tilde{h}||V(\tilde{h})|\leq \frac{Cx^{p+2}}{(1+u)^{\frac{5(2k-D)}{2}}(1+r+u)^{\frac{5(4k-3D+2)-2}{2}}}. 
		\end{align}
		\item Estimate for $H_6$\\
		Again, from (\ref{estimasi Potensial}) we get
		\begin{align}
		\frac{gr}{2}\left|\frac{\partial V(\tilde{h})}{\partial \tilde{h}}\right|\leq \frac{Cx^p}{(1+u)^{\frac{3(2k-D)}{2}}(1+r+u)^{\frac{3(D-2)}{2}-1}}.
		\end{align}
	\end{enumerate}
	
	Finally, we obtain
	\begin{align}
	\left|\frac{\partial h}{\partial u}\right|\leq\frac{C(K_1(x)+K_2(x)+K_3+x^3+x^p+x^{p+2})}{(1+u)^{\frac{3(2k-D)}{2}}(1+r+u)^{\frac{3(D-2)}{2}-1}},
	\end{align}
	where $K_1, K_2,$ and $K_3$ are defined in (\ref{K1})-(\ref{K3}) respectively.
	
	\section*{Acknowledgments}
	The work of this research is supported by PDD Kemendikbudristek 2023,  PPMI FMIPA ITB 2023, PPMI KK ITB 2023, and GTA 50 ITB. 
	
	\section*{Data Availability}
	This manuscript has no associated data.
	
	\section*{Conflict of Interest}
	No conflict of interest in this paper.
	
	\appendix

\end{document}